\documentclass[a4paper, 11pt]{article}

\usepackage[english]{babel}

\usepackage[margin=1in]{geometry}

\usepackage{amsmath}
\usepackage{amssymb}
\usepackage{amsthm}
\usepackage{thmtools}
\usepackage{thm-restate}
\usepackage{graphicx}
\usepackage[colorlinks=true, allcolors=blue]{hyperref}
\usepackage[capitalise]{cleveref}
\usepackage{xspace}
\usepackage[dvipsnames]{xcolor}
\usepackage{tikz}
\usepackage{enumerate}

\usetikzlibrary{calc}

\newtheorem{theorem}{Theorem}
\newtheorem{lemma}[theorem]{Lemma}
\newtheorem{corollary}[theorem]{Corollary}
\newtheorem{observation}[theorem]{Observation}
\newtheorem{claim}[theorem]{Claim}
\theoremstyle{definition}
\newtheorem{definition}[theorem]{Definition}
\theoremstyle{remark}

\newenvironment{claimproof}[1][\proofname]{
  
  \begin{proof}[#1]}{\end{proof}}

\newcommand{\pls}{\textup{\textsf{PLS}}\xspace}
\newcommand{\kopt}{\text{$k$-Opt}\xspace}
\newcommand{\infedge}{\text{non-edge}\xspace}
\newcommand{\fedge}{\text{$G$-edge}\xspace}
\newcommand{\fedges}{\text{$G$-edges}\xspace}

\newcommand{\Z}{\mathbb{Z}}

\definecolor{defblue}{rgb}{0, 0.4, 0.796}
\newcommand{\defi}[1]{\textcolor{defblue}{\emph{#1}}}

\title{On the PLS-Completeness of $k$-Opt Local Search for the Traveling Salesman Problem}
\author{Sophia Heimann\thanks{Hertz Chair for Algorithms and Optimization, University of Bonn, Germany (sheimann@uni-bonn.de)}, 
        Hung P. Hoang\thanks{Algorithms and Complexity Group, Faculty of Informatics, TU Wien, Austria, (phoang@ac.tuwien.ac.at) funded by the Austrian Science Foundation (FWF, projects 10.55776/Y1329 and ESP1136425)}, 
        Stefan Hougardy\thanks{Research Institute for Discrete Mathematics and Hausdorff Center for Mathematics, University of Bonn, Germany (hougardy@or.uni-bonn.de) funded by the Deutsche Forschungsgemeinschaft (DFG, German Research Foundation) under Germany's Excellence Strategy -- EXC-2047/1 -- 390685813}}
\begin{document}
 {\let\thefootnote\relax\footnotetext{A preliminary version of the result was announced at ICALP 2024~\cite{HHH2024}.}}
\maketitle

\begin{abstract}
 The \kopt algorithm is a local search algorithm for the traveling salesman problem.
 Starting with an initial tour, it iteratively replaces at most $k$ edges in the tour with the same number of edges to obtain a better tour.
 Krentel (FOCS 1989) showed that the traveling salesman problem with the \kopt neighborhood is complete for the class \pls (polynomial time local search).
 However, his proof requires $k \gg 1000$ and has a substantial gap.
 We provide the first rigorous proof for the \pls-completeness and at the same time drastically lower the value of $k$ to $k \geq 15$, addressing an open question by Monien, Dumrauf, and Tscheuschner (ICALP 2010).
 Our result holds for both the general and the metric traveling salesman problem.
 \end{abstract}

\section{Introduction}
The well-known Traveling Salesman Problem (TSP) consists of finding a Hamiltonian cycle in an 
edge-weighted complete graph such that the total edge weight of the cycle is the smallest possible.
A popular heuristic for this problem is a local search algorithm called \kopt.
Starting with an arbitrary tour, it iteratively replaces at most $k$ edges in the tour with the same number of edges, as long as the resulting tour has a smaller total edge weight.
We define TSP/\kopt as the problem of finding a local optimum for a TSP instance 
with the \kopt algorithm.
See~\cite{HHH2026} for a summary of related work on this problem and the \kopt algorithm.

A classical method to analyze the complexity of such a local search problem is via the complexity class \pls and the notion of \pls-completeness (for definitions, see \cref{sec:preliminaries}).
They were introduced in 1988 by Johnson, Papadimitriou, and Yannakakis~\cite{JPY1988} 
to capture the observation that for many \textsf{NP}-hard problems it is not only difficult to compute a global optimum, but
even computing a local optimum is also hard. Examples of such problems are the Maximum Satisfiability problem~\cite{krentel1990}, 
Max-Cut~\cite{schaffer1991}, and Set Cover~\cite{DS10}. The \pls-completeness of a problem means that a polynomial time
algorithm to find a local optimum for that problem would imply polynomial time algorithms for finding a local optimum for all problems in \pls. 

The \pls-completeness of TSP/\kopt was proved by Krentel~\cite{Kre1989} for $k \gg 1000$. 
Krentel estimated that his proof yields a value for $k$ between $1{,}000$ and $10{,}000$. 
By using a straightforward way to implement some missing details in Krentel's proof it was recently shown that 
his proof yields the value $14{,}208$ for $k$~\cite{HH2023}.
Further, his proof has a substantial gap 
as he assumes that edges of weight infinity cannot occur in a local optimum (for more details, see \cref{sec:no_infty_opt}). 
Following Krentel's paper, there have been claims in other papers~\cite{JG1997, Yan1997} through private communication with Krentel that a careful analysis of the original proof can bring down the value to $k = 8$ 
and conceivably to $k = 6$. However, there has been no available written proof for these claims.
In fact, to date, the 1989 paper of Krentel~\cite{Kre1989} is the only paper on the topic.
Consequently, Monien, Dumrauf, and Tscheuschner~\cite{Monien_Dumrauf_Tscheuschner_2010} posed an open question on the complexity of TSP/\kopt for $k \ll 1000$.

In this paper, we present in \cref{sec:PLS-complete-proof}
the first rigorous proof for the PLS-completeness of TSP/\kopt and at the same time drastically lower the value of
$k$ from Krentel's  $k \gg 1000$~\cite{Kre1989} to $k \ge 15$. 

\begin{restatable}{theorem}{plscomplete}
\label{thm:PLS_complete}
    TSP/\kopt is \pls-complete for $k \geq 15$.
\end{restatable}

Our proof of \cref{thm:PLS_complete} is based on a new reduction from the bounded degree Max-Cut problem to TSP (see \cref{sec:reduction}). 
To ensure that we achieve a \pls-reduction, we need to take care of the order in which the gadgets are plugged together in our construction. 
We show the correspondence between the solutions of the Max-Cut instance and certain solutions (called \emph{standard tours}) of the constructed graph (see \cref{sec:correspondence}).
Further, we show in \cref{sec:flexible} that with care, we can guarantee that the constructed graph has no other solutions.
We then in~\cref{sec:no_infty_opt} explain how to assign specific weights to the non-edges (i.e., extra edges added to the constructed graph just to make the graph complete) to prove that no local optimum can contain such an edge. 
We achieve this by defining a weight assignment that exploits the special structure of the TSP instance resulting from our \pls-reduction.  
This is the first rigorous proof of such a result for the \kopt algorithm and there seems 
not to be a generic way to prove it for arbitrary TSP instances (as for example those constructed by Krentel~\cite{Kre1989}). 

As observed by Krentel~\cite{Kre1989}, any TSP instance can be made metric by adding a sufficiently large constant to all edge weights. 
This has no effect on the behavior of the \kopt algorithm. 
Hence, \cref{thm:PLS_complete} also holds for Metric TSP, the version of the TSP where distances have to satisfy the triangle inequality.

\paragraph{Comparison with the preliminary conference version.}
We announced a preliminary set of results at ICALP 2024~\cite{HHH2024}, including the \pls-completeness of TSP/\kopt for $k \geq 17$.
In the current paper, we strengthen this result in two ways: we extend it to $k \geq 15$ and present  a tight \pls-reduction (see \cref{def:tight_pls_complete} below), a stronger notion than the original \pls-reduction.
The conference version also includes the \defi{all-exp property} of TSP/\kopt for $k \geq 5$ (i.e., that there exist infinitely many pairs of instances and solutions such that the \kopt algorithm starting from these solutions always takes exponential time).
However, this result was later improved to $k \geq 3$ using a different reduction~\cite{HHH2026}.
Hence, the all-exp property will be discussed in a separate publication.

\section{Preliminaries}
\label{sec:preliminaries}

A $uv$-path is a path that has the vertices $u$ and $v$ as endpoints.

\subsection{Local search problems and the class \pls}
A \defi{local search problem} $P$ is an optimization problem that consists of a set of instances $D_{P}$, a finite set of (feasible) solutions $F_{P}(I)$ for each instance $I\in D_{P}$, an objective function $f_{P}$ that assigns an integer value to each instance $I\in D_{P}$ and solution $s\in F_P(I)$, and a neighborhood $N_{P}(s,I)\subseteq F_{P}(I)$ for each solution $s\in F_{P}(I)$. 
The size of every solution $s \in F_{P}(I)$ is bounded by a polynomial in the size of $I$. 
The goal is to find a \defi{locally optimal solution} for a given instance $I$; that is, a solution $s \in F_{P}(I)$, such that no solution $s' \in N_{P}(s,I)$ yields a better objective value than $f_P(s,I)$.
Formally, this means, for all $s'\in N_{P}(s,I)$, $f_{P}(s,I)\leq f_{P}(s',I)$ if $P$ is a minimization problem and $f_{P}(s,I)\geq f_{P}(s',I)$ if $P$ is a maximization problem.

The \defi{transition graph} $T_I$ of an instance $I$ of a local search problem is a directed graph such that the vertices are the solutions of $I$, and an edge $(s, s')$ exists if and only if $s'$ is a neighbor of $s$ with a better objective value.

We now formalize some classical notions related to the complexity class \pls.

\begin{definition}[The class~\pls~\cite{JPY1988}]
	A local search problem $P$ is in the class \defi{\pls}, if there are three polynomial time algorithms $A_{P}, \ B_{P}, \ C_{P}$ such that 
	\begin{itemize}
		\item Given an instance $I \in D_P$ , $A_{P}$ returns a solution  $s \in F_{P}(I)$;
		\item Given an instance $I \in D_{P}$ and a solution $s\in  F_{P}(I)$, $B_{P}$ computes the objective value $f_{P}(s,I)$ of $s$; and 
		\item Given an instance $I \in D_P$ and a solution $s \in F_{P}(I)$, $C_{P}$ returns a neighbor of $s$ with strictly better objective value, if it exists, and ``locally optimal", otherwise.
	\end{itemize}
\end{definition}

The \defi{standard local search algorithm}~\cite{JPY1988} for an instance $I$ of a problem $P \in \pls$
proceeds as follows.
It starts with some initial solution $s \in F_{P}(I)$ (using algorithm $A_P$).
Then it iteratively visits a neighbor with better objective value (using algorithm $C_P$), until it reaches a local optimum. 
If a solution has more than one better neighbor, the algorithm has to choose one by some prespecified rule, often referred to as a \defi{pivot rule}.

A local search problem~$P$ has the \defi{all-exp} property, if there are infinitely many pairs of an instance~$I$ of $D_P$ and an initial solution $s \in F_P(I)$, for which the standard local search algorithm always needs an exponential number of iterations 
for all possible pivot rules.

\begin{definition}[\pls-reduction~\cite{JPY1988}]
\label{def:pls_complete}
	A \defi{\pls-reduction} from a problem $P \in \pls$ to a problem $Q \in \pls$ is a pair of polynomial-time computable functions $h$ and $g$ that satisfy:
	\begin{enumerate}
		\item Given an instance $I \in D_{P}$, $h$ computes an instance $h(I) \in D_{Q}$; and
		\item Given an instance $I \in D_{P}$ and a solution $s_q \in F_{Q}(h(I))$, $g$ returns a solution $s_p \in F_{P}(I)$ such that if $s_q$ is a local optimum for $h(I)$, then $s_p$ is a local optimum for $I$. 
	\end{enumerate} 
	A problem $Q\in \pls$ is \defi{\pls-complete} if for every problem~$P \in \pls$, there exists a \pls reduction from $P$ to $Q$.
\end{definition}

\begin{definition}[Tight \pls-reduction~\cite{schaffer1991}]
\label{def:tight_pls_complete}
	A \defi{tight \pls-reduction} from a problem $P \in \pls$ to a problem $Q \in \pls$ is a \pls-reduction $(h,g)$ from $P$ to $Q$ such that for every instance $I \in D_P$, we can choose a subset $R$ of $F_Q(h(I))$ that satisfies:
	\begin{enumerate}
		\item $R$ contains all local optima of $h(I)$;
		\item For every solution $s_p \in F_{P}(I)$, we can construct in polynomial time a solution $s_q \in R$ so that $g(I,s_q) = s_p$;
		\item If there is a path from $s_q \in R$ to $s'_q \in R$ in the transition graph $T_{h(I)}$ such that the path contains no other element of $R$, then either there is an edge from the solution $g(I,s_q)$ to the solution $g(I,s'_q)$ in the transition graph $T_I$ or these two solutions are identical.
	\end{enumerate} 
	A problem $Q\in \pls$ is \defi{tightly \pls-complete} if for every problem~$P \in \pls$, there exists a tight \pls-reduction from $P$ to $Q$.
\end{definition}

Sch\"affer and Yannakakis~\cite{schaffer1991} have shown that all tightly PLS-complete problems
have the all-exp property. See the book of Michiels, Aarts, and Korst~\cite{MAK2007} for more background on tightly PLS-complete problems.

\subsection{TSP/\texorpdfstring{\kopt}{k-opt}}
A \defi{Hamiltonian cycle} or a \defi{tour} of an undirected graph is a cycle that contains all vertices of the graph.

A TSP instance is a tuple $(G, c)$, where $G$ is a complete undirected graph $(V, E)$, and $c:E\to \Z_{\geq 0}$ is a function that assigns a nonnegative weight to each edge of $G$.
The goal is to find a tour of $G$ that minimizes the sum of edge weights in the tour.
The definition of the class PLS requires that we have a polynomial time algorithm to find \emph{some}
solution. For complete graphs such an algorithm certainly exists. If the graph is not complete then 
because of the NP-completeness of the Hamiltonian cycle problem we do not know such an algorithm. 

A \defi{swap} is a tuple $(E_1, E_2)$ of subsets $E_1, E_2 \subseteq E$, $|E_1| = |E_2|$.
We say that it is a swap of $|E_1|$ edges.
If $|E_1| \leq k$ for some $k$, then we call it a \defi{$k$-swap}.
Performing a swap $(E_1, E_2)$ from a subgraph $G'$ of $G$ refers to the act of removing $E_1$ from $G'$ and adding $E_2$ to $G'$.
We also call it swapping $E_1$ for $E_2$ in $G'$.
Given a tour~$\tau$, a swap $(E_1, E_2)$ is \defi{improving} for $\tau$, if after swapping $E_1$ for $E_2$ in $\tau$, we obtain a tour with a lower total edge weight.

A \defi{($k$-)swap sequence} is a sequence $L = (S_1, \dots, S_{\ell})$, such that each $S_i$ is a ($k$-)swap.
For a tour $\tau$, we denote by $\tau^L$ the subgraph obtained from $\tau$ by performing $S_1, \dots, S_{\ell}$ in their order in $L$.
$L$ is \defi{improving} for a tour~$\tau$ if each $S_i$ is an improving ($k$-)swap for $\tau^{(S_1, \ldots, S_{i-1})}$.

The local search problem TSP/\kopt corresponds to TSP with the \kopt neighborhood (that is, the neighbors of a tour~$\tau$ are those tours that can be obtained from~$\tau$ by a $k$-swap). 
The \kopt algorithm is then the standard local search algorithm for this problem, and an execution of the algorithm corresponds to an improving $k$-swap sequence.

\subsection{Max-Cut/Flip}
A Max-Cut instance is a tuple $(H, w)$, 
where $H$ is an undirected graph $(V,E)$ and $w: E \to \Z$ is a function that assigns weights to the edges of $H$.
A \defi{cut} $(V_1, V_2)$ of $H$ is a partition of the vertices of $H$ into two disjoint sets $V_1$ and $V_2$, which we call the \defi{first set} and \defi{second set} of the cut, respectively.
The \defi{cut-set} of a cut $(V_1, V_2)$ is the set of edges~$xy \in E$ such that $x \in V_1$ and $y \in V_2$.
The goal of Max-Cut is to find a cut to maximize the \defi{value} of the cut, that is the total weight of the edges in the cut-set.

Given a Max-Cut instance and an initial cut, the \defi{flip} of a vertex is a move of that vertex from a set of the cut to the other.
The flip of a vertex is \defi{improving}, if it results in an increase in the value of the cut.
For a cut~$\sigma$, its \defi{flip neighborhood} is the set of all cuts obtained from~$\sigma$ by an improving flip.
The Max-Cut/Flip problem is the local search problem that corresponds to the Max-Cut problem with the flip neighborhood.
We call its standard local search algorithm the \defi{Flip algorithm}.
A \defi{flip sequence} is a sequence $(v_1, \dots, v_{\ell})$ of vertices of $H$.
A flip sequence is \defi{improving}, if flipping the vertices in the order in the sequence increases the value of the cut at every step.
In other words, an improving flip sequence corresponds to an execution of the Flip algorithm.

Schäffer and Yannakakis~\cite{schaffer1991} proved that Max-Cut/Flip is tightly \pls-complete.
Later, Els\"{a}sser and Tscheuschner~\cite{Elsaesser_Max_Cut_5} showed that Max-Cut/Flip is \pls-complete, even when restricted to graphs of maximum degree five.
However, their reduction is not a tight \pls-reduction, and hence, it is currently not known if Max-Cut/Flip is tightly \pls-complete for constant maximum degree.
Further, note that the Flip algorithm on graphs with maximum degree at most three always terminates after a polynomial number of iterations~\cite{Poljak1995}.
Hence, it is also open whether Max-Cut/Flip is \pls-complete for graphs with maximum degree four; however, it is known that the all-exp property holds in this case~\cite{Monien_Max_Cut_4, Michel_Max_Cut_4}.

\section{The main reduction}
\label{sec:reduction}
Our proof of~\cref{thm:PLS_complete} is based on a new reduction from Max-Cut/Flip
to TSP/\kopt which we will explain in this section. Our reduction will work for all $k \ge 15$. 

Let $(H, w)$ be a Max-Cut instance of maximum degree five.
In order to avoid confusion with the vertices and edges in the TSP instance later on, we use \defi{$H$-vertices} and \defi{$H$-edges} for the vertices and edges of $H$.
We denote by $n$ and $m$ the number of $H$-vertices and $H$-edges, respectively.

We construct from $H$ the corresponding TSP instance as follows.
We start with a cycle of $3(n+m)$ edges.
We assign $n+m$ of the edges, one to each of the $n$ $H$-vertices and the $m$ $H$-edges, such that any two assigned edges have distance at least two on the cycle (see Fig.~\ref{fig:reduction}).

\begin{figure}[ht]
    \centering
    \begin{tikzpicture}[scale=0.75]
    \def\dist{1.35}
    \def\ynodes{-0.75}
    \def\ylabels{-1.75}
    \node[label=above: $x_{\ell}$]    (xl)  at (2*\dist, 0) {};
    \node[label=above: $x^{}_1$] (x1)  at (3*\dist, \ylabels) {};
    \node[label=above: $x_1'$]   (x1') at (4*\dist, \ylabels) {};
    \node[label=above: $x^{}_2$] (x2)  at (5*\dist, \ylabels) {};
    \node[label=above: $x_2'$]   (x2') at (6*\dist, \ylabels) {};
    \node[label=above: $x^{}_3$] (x3)  at (7*\dist, \ylabels) {};
    \node[label=above: $x_3'$]   (x3') at (8*\dist, \ylabels) {};
    \node[label=above: $x_r$]    (xr)  at (9*\dist, 0) {};
    
    \def\eps{0.01}
    \foreach \x in {0,1,2,9,10,11}
        \fill (\x*\dist, 0) circle (0.75mm);
    \foreach \x in {3,...,8}
        \fill (\x*\dist, \ynodes) circle (0.75mm);
    \draw[black, line width = 0.2mm] (0,0)        -- ((11*\dist, 0) 
                                     (2*\dist, 0) -- (3*\dist, \ynodes)
                                     (4*\dist, \ynodes) -- (5*\dist, \ynodes)
                                     (6*\dist, \ynodes) -- (7*\dist, \ynodes)
                                     (8*\dist, \ynodes) -- (9*\dist, 0);  
    \draw[black, line width = 0.25mm, dashed] (3*\dist, \ynodes) -- (4*\dist, \ynodes)
                                              (5*\dist, \ynodes) -- (6*\dist, \ynodes)
                                              (7*\dist, \ynodes) -- (8*\dist, \ynodes);  
    \end{tikzpicture}
    \caption{The first-set edge $x_{\ell} x_r$ and the second-set path $(x_{\ell}, x_1, x'_1, x_2, x'_2, x_3, x'_3, x_r)$ of an $H$-vertex $x$ of degree three. The dashed edges are gateways. The other edges of the second-set path are doors.}
    \label{fig:vertex_gadget}
\end{figure}
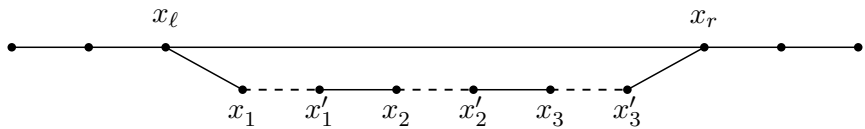

Next, in the cycle consider an edge that is assigned to an $H$-vertex $x$.
(Refer to \cref{fig:vertex_gadget} for an illustration of the following concepts.)
We label the two incident vertices of this edge $x_{\ell}$ and $x_r$, representing the left and the right vertex of the edge.
Let $d(x)$ be the degree of $x$ in $H$.
We add a new path of length $2d(x)+1$ to connect $x_{\ell}$ and $x_r$.
We call this new path the \defi{second-set path} of $x$, while we call the original edge that was assigned to $x$ the \defi{first-set edge} of $x$.
The idea is that the tour can connect $x_{\ell}$ and $x_r$ either via the first-set edge or via the second-set path.
This simulates whether the $H$-vertex $x$ is in the first set or second set of the cut for the Max-Cut problem.
Let $x_{\ell}, x_{1}, x'_{1}, \dots, x_{d(x)}, x'_{d(x)}, x_r$ be the labels of the vertices along the second-set path.
For $i \in \{1, \dots, d(x)\}$, we call the edge $x_i x'_i$ a \defi{gateway} of $x$.
The other edges of the second-set path are called the \defi{doors} of $x$.
In other words, we have alternating doors and gateways along the second-set path, with doors at both ends of the path.

For each $H$-edge $xy$, we call the edge in the cycle of length $3(n+m)$  assigned to $xy$ the 
\defi{$xy$-edge}.
We remove a gateway of $x$, a gateway of $y$, and the $xy$-edge, and we connect the six incident vertices of the three removed edges by a \emph{parity gadget}. 

The purpose of this parity gadget is to simulate the contribution of the weight of edge $xy$ to the objective of the Max-Cut problem, based on whether $x$ and $y$ are in the same set. We will formally define parity gadgets in \cref{sec:parity-gadget}. There we will also precisely describe how a parity gadget will be connected to the rest of the graph. We call this construction step the 
\emph{assigning} of a parity gadget. 

Finally, for each $H$-vertex~$x$, we assign an \emph{XOR gadget} to the first-set edge of~$x$ and the door of $x$ incident to $x_r$.
The purpose of the XOR gadget is to make sure that we can simulate only one flip in $H$ by a $k$-swap in the new graph.
The formal definitions of the XOR gadget and its assignment are discussed in \cref{sec:xor_gadgets}.

\begin{figure}[ht!]
    \centering
    \includegraphics[width=\hsize]{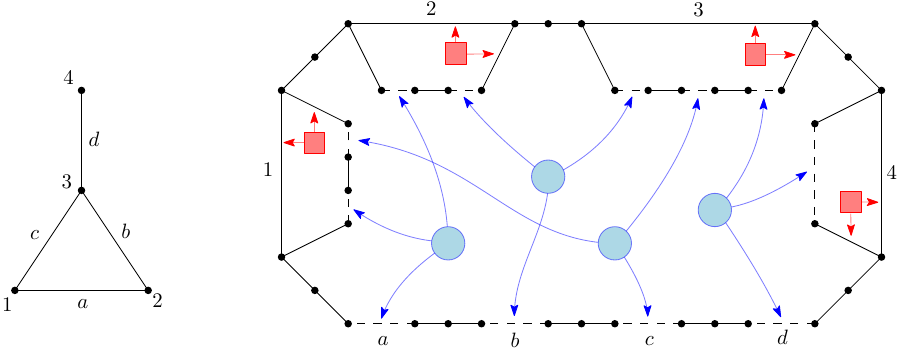}
    \caption{An example of our reduction from a Max-Cut instance (left figure) to a TSP instance (right figure). The parity gadgets are indicated by the blue circles attached to three edges each.
    The XOR gadgets are indicated by red boxes attached to two edges each.}
    \label{fig:reduction}
\end{figure}

The details of how we assign the parity and XOR gadgets are discussed in \cref{sec:equip}.
Except for certain edges in the parity gadgets, which we will specify later, the other edges have weight zero, including the edges in the XOR gadgets, the initial cycle, and the doors.
Let $G$ and $c$ be the resulting graph and weight after all the operations above (see \cref{fig:reduction} for an example).

As a TSP instance requires a complete graph, we add the remaining edges with some suitable large weights, which will be specified in \cref{sec:no_infty_opt}, to obtain the final TSP instance~$(G_{\infty}, c_{\infty})$.
We call these additional edges \defi{non-edges}. 
Our choice of weights for the non-edges will guarantee that  
if we start with a tour that has all edges in $G$, then the \kopt algorithm will never
return a tour that contains a non-edge.

\subsection{Parity gadgets}
\label{sec:parity-gadget}
We now define the parity gadgets, that will have two variants: the \defi{strict gadget} and the \defi{flexible gadget}, defined as the graphs in Figures~\ref{fig:strict_gadget_1} and \ref{fig:flexible_gadget}, respectively.
First we describe some properties and terminologies of these gadgets.

\pgfdeclarelayer{back}
\pgfsetlayers{back,main}

	\tikzset{redline/.style ={draw=red, line width = 1.613}}
	\tikzset{blackline/.style ={draw=black, line width = 0.733}}
	\tikzset{reddashedline/.style ={draw=red, dashed, line width = 0.3666}}
	\tikzset{blackdashedline/.style ={draw=black, dashed, line width = 0.3666}}
	\tikzstyle{vertex}=[black,circle,fill,minimum size=3.666pt, inner sep = 0pt, outer sep = -1] 

\def\nodesnadfixededgesofparitygadget{
    \node[vertex] (a') at (0,4) {};
    \node[vertex] (a) at (4,4) {};
    \node[vertex, label=below: $X$] (X) at (2,2) {}; 
    \node[vertex, label={[shift={(0,0.15)}]right: $X'$}] (X') at (2,6) {};
    \node[vertex] (b') at (1,5) {};
    \node[vertex] (b) at (3,3) {};
    \node[vertex, label=below left: $Z$] (Z) at (4,0) {};
    \node[vertex, label=below right: $Z'$] (Z') at (7,0) {};
    \node[vertex, label=above: $Y'$] (Y') at (7,4) {};
    \node[vertex, label=above: $Y$] (Y) at (9,2) {};
    \node (oZ) at (4,-2.2) {};
    \node (oZ') at (7,-2.2) {};
    \node (oY) at (10.5, 2) {};
    \node (oY') at (10.5,4) {};
    \node (oX) at (-1, 2) {};
    \node (oX') at (-1,6) {};

    \begin{pgfonlayer}{back}
    \draw[redline] (a') -- (X) node[vertex, midway] {};
    \draw[redline] (a) -- (X') node[vertex, midway] {};
    \draw[redline] (b) -- (b') node[vertex, midway] {};
    \draw[redline] (Y) -- (Y') node[vertex, midway] {};
    \draw[reddashedline] (oZ) -- (Z) (oZ') -- (Z');
    
    \end{pgfonlayer}

    \coordinate[label=below: $\sigma$] (sigma) at ($(Z)!1/2!(Z')$);
    \coordinate[label=above: $\sigma$] (sigma) at ($(Y')!1/2!(a)$);
    \coordinate[label=above: $\delta~$] (delta) at ($(Y')!3/4!(Z)$);
    \coordinate[label=above: $~~\delta$] (delta) at ($(a)!3/4!(Z')$);
}    

\begin{figure}[ht]
	\centering

\begin{tikzpicture}[scale=0.4] 
\begin{scope}[shift={(0,0)}] 
    \nodesnadfixededgesofparitygadget
    \begin{pgfonlayer}{back}
    \draw[blackline] (b) -- (a) -- (Z') -- (Z) -- (Y');
    \draw[blackline] (X) -- (b) (Y')--(a) (X') -- (b');
    \draw[blackline] (b') -- (a');
    \draw[blackline] (Y) -- (Z');
    \draw[blackline, out = 180, in =  -90] (Z) to  (a');
    \draw[blackdashedline] (oY) -- (Y) (oY') -- (Y') (oX) -- (X) (oX')--(X');
    \coordinate[label=above: $a$] (sigma) at (a);
    \coordinate[label=below: $b$] (sigma) at (b);
    \coordinate[label=left: $a'$] (sigma) at (a');
    \coordinate[label=left: $b'$] (sigma) at (b');
    \end{pgfonlayer}
 \end{scope}

\begin{scope}[shift={(13,0)}] 
    \nodesnadfixededgesofparitygadget
    \begin{pgfonlayer}{back}
    \draw[blackline] (b) -- (a) -- (Z') -- (Z) -- (Y');
    \draw[redline] (X) -- (b) (Y')--(a) (X') -- (b');
    \draw[blackline] (b') -- (a');
    \draw[redline] (Y) -- (Z');
    \draw[redline, out = 180, in =  -90] (Z) to  (a');
    \draw[blackdashedline] (oY) -- (Y) (oY') -- (Y') (oX) -- (X) (oX')--(X');
    \draw (5.5, -3) node {(1)};
    \end{pgfonlayer}
 \end{scope}
 
\begin{scope}[shift={(26,0)}] 
    \nodesnadfixededgesofparitygadget
    \begin{pgfonlayer}{back}
    \draw[blackline] (X) -- (b) (a) -- (Z') -- (Z)  (Y') -- (a);
    \draw[redline] (a') -- (b') (b) -- (a)  (Z')  (Z) -- (Y') (a);
    \draw[blackline] (X') -- (b');
    \draw[redline] (Y) -- (Z');
    \draw[blackline, out = 180, in =  -90] (Z) to  (a');
    \draw[blackdashedline] (oY) -- (Y) (oY') -- (Y');
    \draw[reddashedline] (oX) -- (X) (oX')--(X');
    \draw (5.5, -3) node {(2)};
    \end{pgfonlayer}
  \end{scope}
\begin{scope}[shift={(6.5,-12)}] 
    \nodesnadfixededgesofparitygadget
    \begin{pgfonlayer}{back}
    \draw[blackline] (b) -- (a) (Y) -- (Z') -- (Z) -- (Y');
    \draw[redline] (X) -- (b)  (X') -- (b') (a) -- (Z');
    \draw[blackline] (b') -- (a') (Y')--(a);
    \draw[redline, out = 180, in =  -90] (Z) to  (a');
    \draw[blackdashedline] (oX) -- (X) (oX')--(X');
    \draw[reddashedline] (oY) -- (Y) (oY') -- (Y');
    \draw (5.5, -3) node {(3)};
    \end{pgfonlayer}
 \end{scope}
\begin{scope}[shift={(19.5,-12)}] 
    \nodesnadfixededgesofparitygadget
    \begin{pgfonlayer}{back}
    \draw[blackline] (X) -- (b) (a) -- (Z') (Z) -- (Y') -- (a);
    \draw[redline] (a') -- (b') (b) -- (a) (Z') -- (Z);
    \draw[blackline] (Y) -- (Z') (X') -- (b');
    \draw[blackline, out = 180, in =  -90] (Z) to  (a');
    \draw[reddashedline] (oY) -- (Y) (oY') -- (Y');
    \draw[reddashedline] (oX) -- (X) (oX')--(X');
    \draw (5.5, -3) node {(4)};
    \end{pgfonlayer}
  \end{scope}
  \end{tikzpicture}
    \caption{The strict gadget (top left) and the subtours (1)--(4). Bold red edges in the top left are edges always in a subtour; those edges in the other panels show the subtours. Dashed edges are external edges (i.e. edges that are incident to the gadget in the overall construction, but not part of the gadget itself). The edges $Y'a$ and $ZZ'$ are same-set edges with same-set weight $\sigma$, while $Z'a$ and $Y'Z$ are different-set edges with different-set weight $\delta$.}
    \label{fig:strict_gadget_1}
\end{figure}
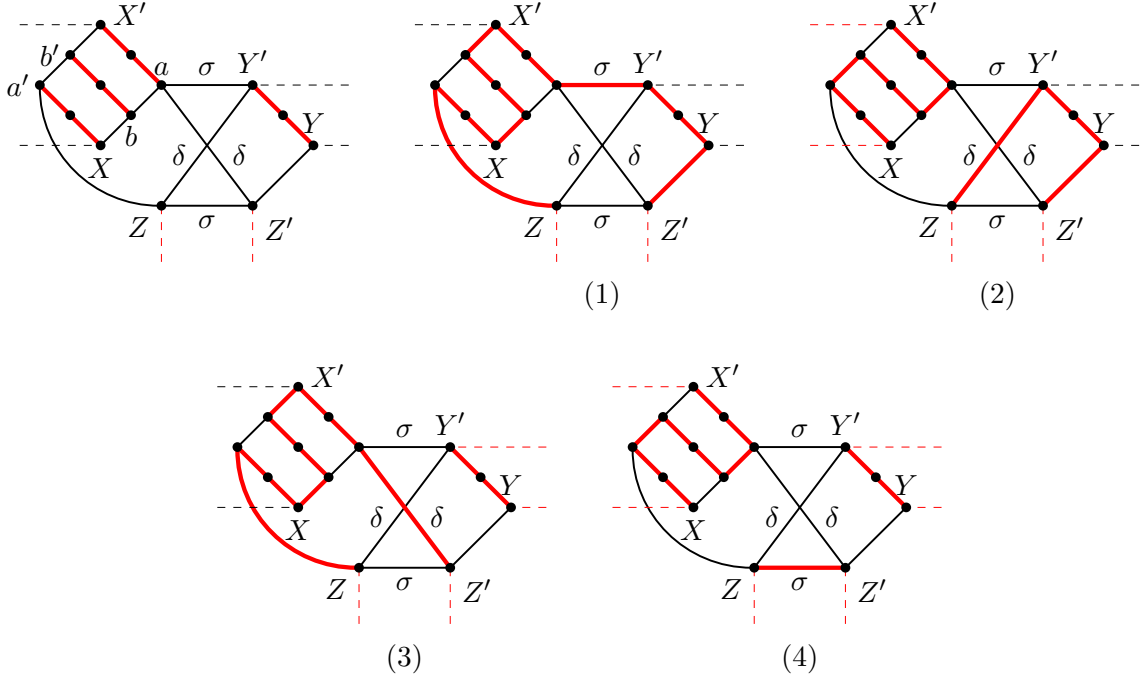

\begin{figure}[ht]
    \centering

\begin{tikzpicture}[scale=0.666]
\small

\newcommand{\vrad}{0.8mm}

\newcommand{\commonpart}{%
  \coordinate[label=left: $X$]  (X)  at (0,3);
  \coordinate[label=left: $Y$]  (Y)  at (0,0);
  \coordinate[label=left: $Z$]  (Z)  at (0,2);
  \coordinate[label=right:$X'$] (Xp) at (1,2);
  \coordinate[label=right:$Y'$] (Yp) at (1,1);
  \coordinate[label=left: $Z'$] (Zp) at (0,1);
  \coordinate (R) at (1,3);
  \coordinate (S) at (1,0);
  \draw[thick, red] (X)--(R)--(Xp) (Yp)--(S)--(Y);
}

\newcommand{\fillverts}[2][]{%
  \foreach \v in {#2} { \fill[#1] (\v) circle (\vrad); }%
}

\commonpart
\draw[thick]  (Y) -- (Zp) (Z) -- (X);
\draw[thick, dashed] (Zp) -- (Z)  (Xp) -- (Yp);
\draw[thick, dotted] (Yp) -- (Z) (Xp)--(Zp);
\fillverts{X,Xp,Y,Yp,Z,Zp,R,S}

\begin{scope}[shift={(3,0)}]
  \commonpart
  \draw[thick] (Yp)  (S)  (Y) (Zp) (Z) (X) (R) (Xp);
  \draw[thick, dotted] (Yp) -- (Z) (Xp)--(Zp);
  \draw[thick, dashed] (Zp) -- (Z);
  \draw[thick, red] (Z) -- (X) (Y)--(Zp);
  \draw[thick, dashed, red] (Xp) -- (Yp);
  \draw (0.5, -0.8) node {(1)};
  \fillverts{X,Xp,Y,Yp,R,S}
  \fillverts[red]{Z,Zp}
\end{scope}

\begin{scope}[shift={(6,0)}]
  \commonpart
  \draw[thick] (X) -- (Z) (Zp)  (Yp) (Xp);
  \draw[thick, dotted] (Xp)--(Zp);
  \draw[thick, dashed] (Zp) -- (Z)  (Xp) -- (Yp);
  \draw[thick, red] (Y)--(Zp);
  \draw[thick, dotted, red] (Yp) -- (Z);
  \draw (0.5, -0.8) node {(2)};
  \fillverts[red]{X,Xp,Z,Zp}
  \fillverts{Y,Yp,R,S}
\end{scope}

\begin{scope}[shift={(9,0)}]
  \commonpart
  \draw[thick] (Z) (Zp) -- (Y)  (Yp)  (Xp);
  \draw[thick, dotted] (Yp) -- (Z);
  \draw[thick, dashed] (Zp) -- (Z)  (Xp) -- (Yp);
  \draw[thick, red] (Z)--(X);
  \draw[thick, dotted, red] (Xp)--(Zp);
  \draw (0.5, -0.8) node {(3)};
  \fillverts{X,Xp,R,S}
  \fillverts[red]{Y,Yp,Z,Zp}
\end{scope}

\begin{scope}[shift={(12,0)}]
  \commonpart
  \draw[thick] (Yp)  (Xp) (Zp) -- (Y)  (X) -- (Z);
  \draw[thick, dotted] (Yp) -- (Z) (Xp)--(Zp);
  \draw[thick, dashed] (Xp) -- (Yp);
  \draw[thick, dashed, red] (Zp) -- (Z);
  \draw (0.5, -0.8) node {(4)};
  \fillverts[red]{X,Xp,Y,Yp,Z,Zp}
  \fillverts{R,S}
\end{scope}

\begin{scope}[shift={(15,0)}]
  \commonpart
  \draw[thick] (Yp)  (S)  (Y) (Zp) (Z) (X) (R) (Xp)  (Z) -- (X)  (Y)--(Zp);
  \draw[thick, dotted, red] (Yp) -- (Z) (Xp)--(Zp);
  \draw[thick, dashed] (Zp) -- (Z)   (Xp) -- (Yp);
  \fillverts[red]{X,Y,Z,Zp}
  \fillverts{Xp,Yp,R,S}
\end{scope}

\begin{scope}[shift={(18,0)}]
  \commonpart
  \draw[thick] (X) -- (Z) (Zp)  (Yp) (Xp) (Zp) -- (Y);
  \draw[thick, dotted] (Xp)--(Zp) (Yp) -- (Z);
  \draw[thick, dashed, red] (Zp) -- (Z) (Xp) -- (Yp);
  \fillverts[red]{X,Y,Z,Zp}
  \fillverts{Xp,Yp,R,S}
\end{scope}

\begin{scope}[shift={(21,0)}]
  \commonpart
  \draw[thick, dotted] (Xp)--(Zp) (Yp) -- (Z) (Zp) -- (Z) (Xp) -- (Yp) ;
  \draw[thick, red] (X) -- (Z) (Zp)  (Yp) (Xp) (Zp) -- (Y);
  \fillverts[red]{Xp,Yp,Z,Zp}
  \fillverts{X,Y,R,S}
\end{scope}

\end{tikzpicture}

    \caption{The flexible gadget (left hand figure). The red edges in the left figure show the four edges that have to be contained in every subtour. The next four figures show the four possibilities for subtour~(1) -- subtour~(4) to cover the vertices of the flexible gadget by disjoint paths (red edges and red endpoints). The last three figures are the three non-standard subtours. The dashed edges are the same-set edges, the dotted edges are the different-set edges, and the solid edges are the remaining edges.}
    \label{fig:flexible_gadget}
\end{figure}
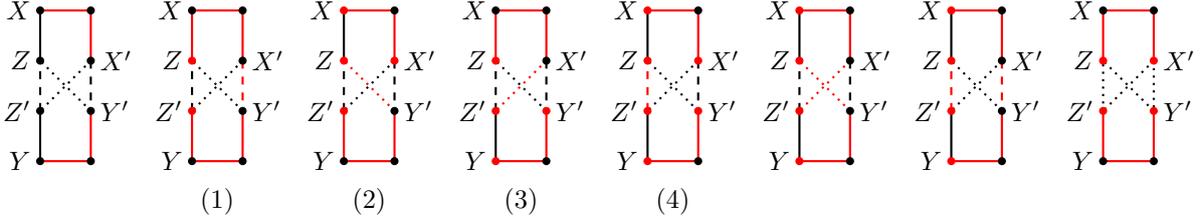

A parity gadget contains six distinct vertices labeled $X, X', Y, Y', Z, Z'$.
A \defi{subtour} of the parity gadget is a spanning subgraph $T$ such that
\begin{itemize}
    \item $T$ is a set of vertex disjoint paths of length at least one,
    \item $Z$ and $Z'$ are endpoints of some path(s) above, and
    \item The other endpoints of the paths are in the set $\{X,X', Y, Y'\}$.
\end{itemize}
Figures~\ref{fig:strict_gadget_1} and \ref{fig:flexible_gadget} show the existence of the following four subtours:
\begin{enumerate}
    \item[(1)] A $ZZ'$-path;
    \item[(2)] An $XX'$-path and a $ZZ'$-path;
    \item[(3)] A $YY'$-path and a $ZZ'$-path; or
    \item[(4)] An $XX'$-path, a $YY'$-path and a $ZZ'$-path.
\end{enumerate}
We call them \defi{subtour (1)}, \defi{subtour (2)}, \defi{subtour (3)}, and \defi{subtour (4)}, respectively.
They are the \defi{standard subtours}, and all other subtours are \defi{non-standard}.

Figures~\ref{fig:strict_gadget_1} and \ref{fig:flexible_gadget} also show the partition of the edges into three subsets, the \defi{same-set edges}, the \defi{different-set edges}, and the remaining edges.
The same-set edges have the same weight, which we call the \defi{same-set weight}.
Similarly, the different-set edges have the same \defi{different-set weight}.
The remaining edges have weight zero. 

\begin{lemma}
\label{lem:unique_subtours}
	Firstly, for $i \in \{1, 2, 3, 4\}$, every parity gadget has a unique subtour $(i)$.
	Secondly, for every parity gadget, the total weights of subtours (1) and (4) are the same-set weight, while those of subtours (2) and (3) are the different-set weight.
	Lastly, the flexible gadget has exactly three non-standard subtours as depicted in \cref{fig:flexible_gadget},  while the strict gadget does not have any.
\end{lemma}
\begin{proof}
    We start with the strict gadget as depicted in \cref{fig:strict_gadget_1}.
    Note that the eight edges in the strict gadget that are incident to a vertex of degree~2 
    must be contained in every subtour in the strict gadget.
    
    Firstly, suppose that $X$ is an endpoint of a path in a subtour~$\zeta$ of the strict gadget.
    Then~$Xb$ is not in~$\zeta$, and hence~$ab$ is in~$\zeta$.
    This then implies that $a'b'$ is in $\zeta$, because otherwise, we have a cycle in $\zeta$.
    It follows that the edges $Za'$, $Z'a$, and $aY'$ are not in a path from the subtour.
    By definition of a subtour, the vertices $Z$ and $Z'$ must be endpoints of some paths in the subtour. 
    If $ZZ'$ is an edge in the path cover, then we obtain the subtour~(4).
    If $ZZ'$ is not an edge then $ZY'$ and $Z'Y$ must be edges in the subtour. Thus we obtain subtour~(2).

    Secondly, assume that $X$ is not an endpoint of a path in $\zeta$.
    By similar argument as above, we have that $Xb$ and $X'b'$ are in $\zeta$, while $a'b'$ and $ab$ are not.
    This implies that $Za'$ is in the subtour. 
    If $Z'a$ is an edge in the path cover then we obtain subtour~(3). 
    If this is not the case we obtain subtour~(1).

    The two preceding paragraphs imply the non-existence of non-standard subtours
    and the uniqueness of the four standard subtours of the strict gadget which are depicted in \cref{fig:strict_gadget_1}.
    It is easy to verify that the weight of subtours~(1) and~(4) is the same-set weight, while that of subtours~(2) and~(3) is the different-set weight.
    
    Lastly, for the flexible gadget, we can easily verify the lemma statement, using a similar case distinction as above, and observing that $Z$ and $Z'$ have to be endpoints of some paths in the subtour.
\end{proof}

A \defi{change} from a subtour $T_1$ to a subtour $T_2$ is the act of removing $E(T_1) \setminus E(T_2)$ from $T_1$ and then adding $E(T_2) \setminus E(T_1)$.
We then say that the change \defi{involves} $|E(T_1) \setminus E(T_2)| + |E(T_2) \setminus E(T_1)|$ edges.
The uniqueness of the standard subtours in \cref{lem:unique_subtours} above motivates the following definition.

\begin{definition}
\label{def:rx_ry_pg}
    A parity gadget is an \defi{$(r_x,r_y)$-parity gadget}, if
\begin{itemize}
    \item A change either between subtour (1) and subtour (2) or between subtour (3) and subtour (4) involves $2r_x-1$ edges; 
    \item A change either between subtour (1) and subtour (3) or between subtour (2) and subtour (4) involves $2r_y-1$ edges; 
    \item A change either between subtour (1) and subtour (4) or between subtour (2) and subtour (3) involves more than $2\max\{r_x,r_y\}-1$ edges. 
\end{itemize}
\end{definition}

We call the changes in the first two conditions in the definition above the \defi{standard subtour changes}.

From \cref{lem:unique_subtours} and Figures~\ref{fig:strict_gadget_1} and \ref{fig:flexible_gadget}, we have the following observations.

\begin{observation}
\label{ob:num_exchanged_edges}
	The strict gadget is a (4,2)-parity gadget, while the flexible gadget is a (2,2)-parity gadget.
\end{observation}

We now define what we mean by \defi{assigning} a strict/flexible gadget in the graph $G$ to the pair $(x,y)$ for an $H$-edge $xy$: Remove a gateway 
$x_i x'_i$ of $x$ and a gateway 
$y_j y'_j$
of $y$ and the $xy$-edge $\bar{Z}\bar{Z}'$.
Then add a copy of the strict/flexible gadget and identify the vertices $X, X', Y, Y', Z, Z'$ of the copy with the vertices 
$x_i, x_i', y_j, y_j', \bar{Z}, \bar{Z}'$, 
respectively. Recall that $w(xy)$ is the weight of the $H$-edge $xy$.
If $w(xy) \geq 0$, then we set the same-set weight of the parity gadget to $w(xy)$, and its different-set weight to zero. Otherwise, the same-set weight of the parity gadget is set to zero, and its different-set weight is set to $-w(xy)$.

Observe that the assigned parity gadget is connected with the rest of~$G$ via incident edges to~$X$, $X'$, $Y$, $Y'$, $Z$, and $Z'$.
We call these incident edges the \defi{external edges} of the parity gadget.
We define the \defi{internal edges} as the edges within the parity gadget.
Further, we say that the parity gadget is \defi{related} to the $H$-vertices $x$ and $y$.

By construction, the removed edge $\bar{Z}\bar{Z}'$ was originally part of a path of length five, say $(Z_1, Z_2, \bar{Z}, \bar{Z}', Z_3, Z_4)$.
Since $Z_2$ and $Z_3$ have degree two in~$G$, any tour of~$G$ has to contain the edges~$Z_2\bar{Z}$ and~$\bar{Z}'Z_3$.
Therefore, a tour can only contain exactly one internal edge incident to~$Z$ and one incident to~$Z'$ (which may coincide). 
This is the reason why we define a subtour to have $Z$ and $Z'$ as endpoints.

By \cref{lem:unique_subtours}, the total weight of the tour edges within a parity gadget related to the $H$-vertices $x$ and $y$ is the same-set weight, when $x$ and $y$ are in the same set of the cut, and it is the different-set weight, when they are in different sets.

\subsection{XOR gadgets}
\label{sec:xor_gadgets}
We now define the XOR gadgets that we use in our reduction.
We generalize these gadgets from the XOR gadget by~\cite{pap1978}.
See \cref{fig:xor_gadget}(a)-(c) for an illustration of the definition below.

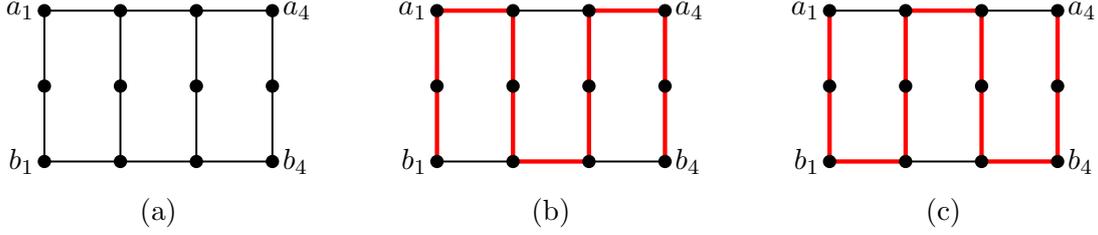
\begin{figure}[t!]
    \centering

    \def\defvertices{    
    \foreach \x in {1,...,4} 
      \foreach  \y in {1,2,3} {
         \fill (\x, \y) circle(2.5pt);
         \coordinate (a\x\y) at (\x,\y);}  
    \coordinate[label=left: $b_1$] (dummy) at (a11);
    \coordinate[label=left: $a_1$] (dummy) at (a13);
    \coordinate[label=right: $b_4$] (dummy) at (a41);
    \coordinate[label=right: $a_4$] (dummy) at (a43);}

    \begin{tabular}{c@{~~~~~~~}c@{~~~~~~~}c}
    \begin{tikzpicture}
    \defvertices
    \foreach \x in {1,...,4} 
      \draw[blackline] (a\x1) -- (a\x3);
 
    \draw[blackline] (a11) -- (a41) (a13) -- (a43);
    \end{tikzpicture}
&
    \begin{tikzpicture}
    \defvertices
    \begin{pgfonlayer}{back}   
    \foreach \x in {1,...,4} 
      \draw[redline] (a\x1) -- (a\x3);
    \draw[redline] (a13) -- (a23) (a33)--(a43) (a21) -- (a31);
    \draw[blackline] (a11) -- (a21) (a31)--(a41) (a23) -- (a33);
    \end{pgfonlayer}
    \end{tikzpicture}
&
    \begin{tikzpicture}
    \defvertices
    \begin{pgfonlayer}{back}   
    \foreach \x in {1,...,4} 
      \draw[redline] (a\x1) -- (a\x3);
    \draw[blackline] (a13) -- (a23) (a33)--(a43) (a21) -- (a31);
    \draw[redline] (a11) -- (a21) (a31)--(a41) (a23) -- (a33);
    \end{pgfonlayer}
    \end{tikzpicture}\\ (a) & (b) & (c)
    \end{tabular}
    \caption{The XOR gadget of order four (a) and its two subtours (b) and (c).}
    \label{fig:xor_gadget}
\end{figure}

\begin{figure}
    \centering
    \begin{tabular}{c@{~~~~~~~}c@{~~~~~~~}c}
    \begin{tikzpicture}
    \coordinate[label=left: $x_l$] (xl) at (0,2);
    \coordinate[label=left: $v$] (v) at (1,0);
    \coordinate[label=right: $x_r$] (xr) at (2,2);
    \coordinate[label=above: \textcolor{white}{$a_1^x$}] (a1) at (2,2);
    \coordinate[label=below: \textcolor{white}{$b_1^x$}] (b1) at (2,0);
    \fill[black] (xl) circle (2.5pt);  
    \fill[black] (v) circle (2.5pt);  
    \fill[black] (xr) circle (2.5pt);  
    \draw[black, dashed, line width = 1] (xl) -- (xr);
    \draw[black, dotted, line width = 1] (v) -- (xr);
    \end{tikzpicture}&
    \begin{tikzpicture}
    \coordinate[label=left: $x_l$] (xl) at (0,2);
    \coordinate[label=left: $v$] (v) at (1,0);
    \coordinate[label=right: $x_r$] (xr) at (3,2);
    \coordinate[label=above: $a_1^x$] (a1) at (2,2);
    \coordinate[label=below: $b_1^x$] (b1) at (2,0);
    \coordinate[] (c1) at (2,1);
    \fill[black] (xl) circle (2.5pt);  
    \fill[black] (v) circle (2.5pt);  
    \fill[black] (xr) circle (2.5pt);  
    \fill[black] (a1) circle (2.5pt);  
    \fill[black] (b1) circle (2.5pt);  
    \fill[black] (c1) circle (2.5pt);  
    \draw[black, dashed, line width = 1] (xl) -- (a1);
    \draw[black, dotted, line width = 1] (v) -- (b1);
    \draw[black, line width = 1] (b1) -- (a1) -- (xr);
    \draw[black, line width = 2.5] (b1) -- (xr);
    \end{tikzpicture}&
    \begin{tikzpicture}
    \coordinate[label=left: $x_l$] (xl) at (0,2);
    \coordinate[label=left: $v$] (v) at (1,0);
    \coordinate[label=right: $x_r$] (xr) at (4,2);
    \coordinate[label=above: $a_1^x$] (a1) at (2,2);
    \coordinate[label=below: $b_1^x$] (b1) at (2,0);
    \coordinate[] (c1) at (2,1);
    \coordinate[label=above: $a_2^x$] (a2) at (3,2);
    \coordinate[label=below: $b_2^x$] (b2) at (3,0);
    \coordinate[] (c2) at (3,1);
    \fill[black] (xl) circle (2.5pt);  
    \fill[black] (v) circle (2.5pt);  
    \fill[black] (xr) circle (2.5pt);  
    \fill[black] (a1) circle (2.5pt);  
    \fill[black] (b1) circle (2.5pt);  
    \fill[black] (c1) circle (2.5pt);  
    \fill[black] (a2) circle (2.5pt);  
    \fill[black] (b2) circle (2.5pt);  
    \fill[black] (c2) circle (2.5pt);  
    \draw[black, dashed, line width = 1] (xl) -- (a1);
    \draw[black, dotted, line width = 1] (v) -- (b1);
    \draw[black, line width = 1] (b2) -- (a2) -- (a1) -- (b1) -- (b2) -- (xr);
    \draw[black, line width = 2.5] (a2) -- (xr);
    \end{tikzpicture}\\ $p=0$ & $p=1$ & $p=2$
    \end{tabular}
    \caption{Examples of assigning an XOR gadget of order 0, 1, and 2 to an $H$-vertex~$x$.
    Dashed edges,        dotted edges,                   and bold edges represent the 
    left first-set edge, the door closest to $x_r$, and the right first-set edge, respectively.}
    \label{fig:assigned_xor_gadget}
\end{figure}
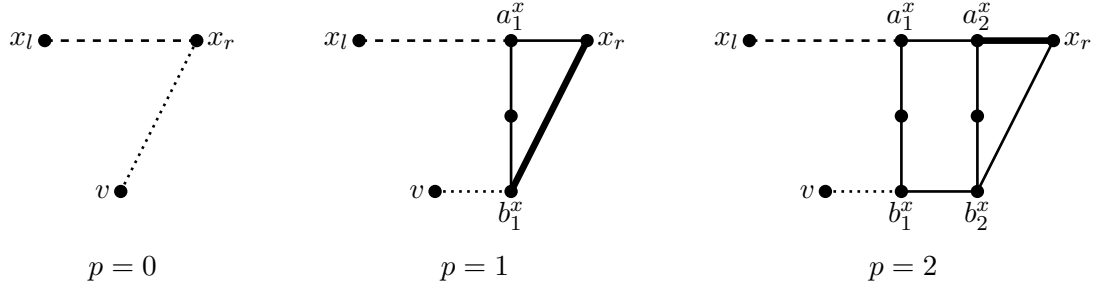

\begin{definition}[XOR Gadget]
\label{def:xor_gadget}
    Let $p$ be a nonnegative integer.
    The \defi{XOR~gadget} of order~$p$ is a graph containing two paths $(a_1, \dots, a_{p})$ and $(b_1, \dots, b_{p})$ which are called the \defi{rails} of the XOR~gadget.
    In addition, for $i \in \{1, \dots, p\}$, the XOR~gadget contains vertex disjoint paths of length two with $a_i$ and $b_i$ as endpoints.  A \defi{subtour} of the XOR~gadget is a spanning path with two endpoints in the set $\{a_1, a_p, b_1, b_p\}$.
    For convenience, when $p = 0$, both the XOR~gadget of order zero and its only subtour are defined to be the empty graph.
\end{definition}

It is easy to see that an XOR~gadget has exactly two subtours, except when $p \leq 1$, in which case, it has only one subtour.
Further, for $p \geq 2$, changing from one subtour to the other requires a swap of $p-1$ edges.

We define the \defi{assigning} of an XOR~gadget of some order~$p$ to an $H$-vertex~$x$ as follows. 
(See~\cref{fig:assigned_xor_gadget} for an illustration.)
Recall that $x_{\ell}x_r$ is the first-set edge of $x$, and let $vx_r$ be the incident door of~$x$ to $x_r$.
We subdivide these two edges into paths of length~$p+1$, $(x_{\ell}, a^x_1, \dots, a^x_p, x_r)$ and $(v, b^x_1, \dots, b^x_p, x_r)$.
Then for $i \in \{1, \dots, p\}$, we connect $a^x_i$ and $b^x_i$ with a path of length two.
Note that when $p = 0$, we do nothing.
Further note that when we remove the edges incident to $x_{\ell}$, $x_r$, and $v$ in the construction above, we obtain the XOR gadget of order~$p$ as defined in \cref{def:xor_gadget}.

We call these incident edges to $x_{\ell}$, $x_r$, and $v$ the \defi{external edges} of the XOR~gadget, and we call the other edges in the construction the \defi{internal edges} of the XOR~gadget.
For convenience, we still refer to the external edge incident to $v$ (i.e., $vb^x_1$ for $p \geq 1$ and $vx_r$ for $p=0$) as a door of $x$.
Additionally, we call it the \defi{closest door to $x_r$}.
We call the external edge incident to $x_{\ell}$ (i.e., $x_{\ell}a^x_1$ for $p \geq 1$ and $x_{\ell}x_r$ for $p = 0$) the \defi{left first-set edge} of~$x$. 
We define the \defi{right first-set edge} of~$x$ as $x_rx_{\ell}$ if $p = 0$, $x_ra^x_p$ if $p$ is positive and even, and $x_rb^x_p$ if $p$ is odd.
Lastly, we call the other external edge incident to $x_r$ the \defi{right second-set edge} of~$x$.
See \cref{fig:assigned_xor_gadget} for examples of the concepts above.

We define an \defi{incident edge} of a nonempty subtour in the XOR gadget to be an external edge incident to an endpoint of the subtour.
The incident edge of an empty subtour (i.e., when $p = 0$) is defined to be simply an external edge of the XOR gadget.

Based on the definitions above, it is easy to verify the following.
\begin{observation}
\label{obs:xor_external_edges}
    For the XOR gadget assigned to an $H$-vertex~$x$, one subtour in the gadget is incident to the left and right first-set edges of~$x$, and another subtour in the gadget is incident to the closest door to~$x_r$ and the right second-set edge of~$x$.
    These two subtours are identical, if the order of the gadget is at most one.
    Otherwise, they are distinct.
\end{observation}

\subsection{Assigning gadgets}
\label{sec:equip}
Within our reduction we use different parity gadgets at different places and XOR gadgets of different orders. 
Now, we specify which parity gadget we assign to which $H$-edge and the XOR gadget of which order to which $H$-vertex. For this we will make use of the following lemma about partial edge orientations. 
Note that the graph $H$ has maximum degree five by assumption.

\begin{lemma} Let $F=(V,E)$ be a graph of maximum degree~5. Then we can compute in polynomial 
time a partition of $E$ into two disjoint sets $E_1$ and $E_2$ and an orientation $\overrightarrow{E_2}$ of the edges in $E_2$ such that $(V, E_1)$ is a disjoint union of stars and each vertex in $(V, \overrightarrow{E_2})$ has out-degree at most~two.   
\label{lem:partialedgeorientation}
\end{lemma}

\begin{proof}
  Let $M$ be a maximum cardinality matching of $F$, which can be computed in polynomial time.
  We initialize $E_1$ to be the edges of $M$.
  Then for every degree-five vertex $v$ of $F$ that is unmatched in $M$, we add an arbitrary incident edge of $v$ to $E_1$.
  We claim that $E_1$ satisfies the lemma condition.
  Indeed, each edge~$e$ in $E_1$ but not in $M$ must be incident to a matched vertex in $M$, because otherwise, we can add $e$ to $M$ to form a larger matching.
  Further, for every edge $uv$ of $M$, if $u$ and $v$ are adjacent to unmatched vertices $u'$ and $v'$ in $F$, respectively, then by removing $uv$ from $M$ and adding $uu'$ and $vv'$, we obtain a larger matching than $M$.
  Hence, at least one of $u$ and $v$ is incident to only matched vertices in $F$.
  This implies that all paths in $(V, E_1)$ have length at most two.
  In other words, it is a disjoint union of stars.
  
  Let $E_2 := E \setminus E_1$.
  The graph $(V, E_2)$ has maximum degree four as $F$ has maximum degree~five and each vertex of degree~five in $F$ has at least one incident edge in $E_1$. 
  Now orient the edges in $E_2$ as follows: Partition $(V, E_2)$ into maximal (possibly closed) walks and 
  orient the edges along each walk. The maximality of the walks guarantees that each vertex has 
  out-degree (and in-degree) at most~two with respect to this orientation.
\end{proof}

  We apply \cref{lem:partialedgeorientation} to the graph $H$ and denote by
  $\overrightarrow{H}$ the resulting graph.
  Note that $\overrightarrow{H}$ has a mix of undirected and directed edges.

    We then assign the parity gadgets based on the orientation of the edges in $\overrightarrow{H}$.
    Let $\psi$ be an arbitrary order of the $H$-edges such that the directed edges precede the undirected edges.
    We process the $H$-edges in this order. 
    Let $xy$ be the current $H$-edge.
    If it is directed in $\overrightarrow{H}$, without loss of generality, suppose it is directed from $x$ to $y$ in $\overrightarrow{H}$.
    We then assign the strict gadget to $(x,y)$.
    If it is undirected in $\overrightarrow{H}$, we assign the flexible gadget to $(x,y)$.
    In either assignment, we remove the available gateways closest to $x_{\ell}$ and $y_{\ell}$.
    
    As each parity gadget is assigned to some $H$-edge, the ordering $\psi$ on the $H$-edges 
    implies a corresponding ordering on the assigned parity gadgets. 
    With a slight abuse of notation, we also use $\psi$ to denote this order of the assigned parity gadgets. We will rely on this ordering $\psi$ in \cref{sec:no_infty_opt}.
    Note that our assignment ensures that for every $H$-vertex $x$, when we go along the second-set path from $x_{\ell}$ to $x_r$, we encounter the gadgets in increasing $\psi$-order.
     
	Finally, we assign the XOR gadgets.
	Recall that $d(x) \leq 5$ denotes the degree of an $H$-vertex $x$.
	Let $d^+(x)$ denote the out-degree of a vertex $x$ in $\overrightarrow{H}$.
	(Note that undirected incident edges do not contribute to this out-degree.)
	By \cref{lem:partialedgeorientation}, every $H$-vertex $x$ has $d^+(x) \leq 2$ .
	Now, to every $H$-vertex $x$, we assign the XOR gadget of order $k - 2d(x) - 2d^+(x) - 1$.
	Since $k \geq 15$, $d(x) \leq 5$, and $d^+(x) \leq 2$, note that this is a valid assignment.  

    See \cref{fig:reduction_full} for an example.

    \begin{figure}[ht!]
    \centering
    \includegraphics[width=\hsize]{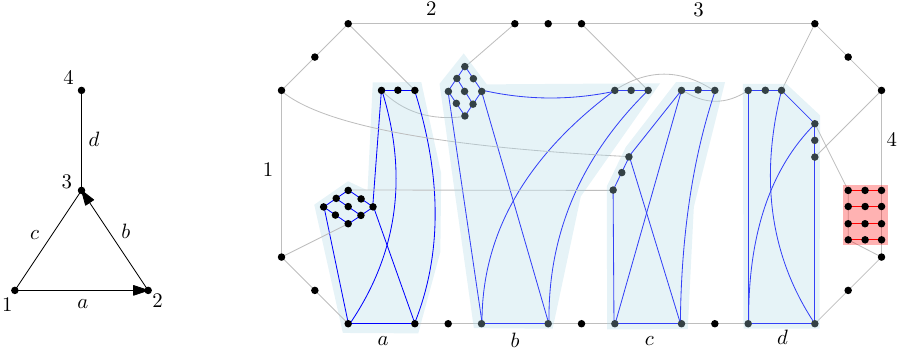}
    \caption{Following the example in \cref{fig:reduction}, suppose $\overrightarrow{H}$ is depicted on the left.
    Then for $k = 7$, we have the construction on the right, where the blue and red subgraphs correspond to the parity gadgets and XOR gadget, respectively.
    Note that the vertices 1--3 are assigned XOR gadgets of order zero.}
    \label{fig:reduction_full}
\end{figure}

\section{Standard tours}
\label{sec:correspondence}
Let $\tau$ be a subgraph of $G$.
For a gadget $\alpha$, the \defi{subtour} of $\tau$ in $\alpha$  is the subgraph of $\alpha$ that contains all common edges of $\alpha$ and $\tau$, i.e., it is the graph $(V(\alpha), E(\alpha) \cap E(\tau))$. We say that $\tau$ is \defi{standard}, if the subtour of $\tau$ in every gadget is a standard subtour of the corresponding gadget.

We start by observing in the following lemma that when $\tau$ is a tour of $G$, then the definition of subtour of $\tau$ in a gadget is consistent with the definition of a subtour of a gadget defined in the previous section.

\begin{lemma}
    Let $\tau$ be a tour of $G$ and $\alpha$ be a gadget in $G$.
    Then a subtour of $\tau$ in $\alpha$ is a subtour of $\alpha$.
\end{lemma}
\begin{proof}
    If $\alpha$ is a parity gadget, by construction, $\alpha$ is connected with the rest of $G$ via the vertices $X, X', Y, Y', Z, Z'$ of $\alpha$.
    Further, $Z$ and $Z'$ are connected to degree-two vertices outside of the gadget.
    It is then easy to see that the subtour of $\tau$ in $\alpha$ must satisfy the definition of a subtour of $\alpha$.
    A similar argument also applies for the case when $\alpha$ is an XOR-gadget of order other than two.

    The remaining case is when $\alpha$ is an XOR-gadget of order two; that is, a six cycle of the form $(a_1, a_2, u_2, b_2, b_1, u_1)$.
    The only possibility for a subtour of $\tau$ in $\alpha$ that is not a subtour of $\alpha$ is the subgraph that consists of two paths $(a_1, u_1, b_1)$ and $(a_2, u_2, b_2)$.
    However, by construction, $a_2$ and $b_2$ are also adjacent to the vertex $x_r$ for some $H$-vertex $x$ and have no other neighbors in $G$ outside of $\alpha$.
    Since $a_2$ and $b_2$ must have degree two in a tour, $\tau$ must contain the edges $a_2x_r$ and $b_2x_r$.
    This means there is a four cycle $(a_2, u_2, b_2, x_r)$ in $\tau$, a contradiction to the fact that $\tau$ is a tour of $G$.
    Therefore this possibility cannot occur.
    The lemma then follows.
\end{proof}

Now the idea of our tight \pls-reduction from Max-Cut/Flip to TSP/\kopt is that we will use the set of all standard tours of the graph $G$ as the set $R$ required in \cref{def:tight_pls_complete}.
In this section, we establish some useful properties of these standard tours. We will need
these properties to prove the tightness of our \pls-reduction.

We start with a simple lemma about doors of an $H$-vertex in an arbitrary tour.

\begin{lemma}
\label{lem:end_doors}
    For any tour of $G$ and any $H$-vertex~$x$, the tour contains either both the door incident to $x_{\ell}$ and the closest door to ${x_r}$ or none of them.
    Further, these two doors are in the tour, if and only if the left and right first-set edges of $x$ are not in the tour.
\end{lemma}
\begin{proof}
    By construction, $x_{\ell}$ is incident to three edges in $G$: the left first-set edge, a door, and an edge incident to a degree-two vertex.
    Since the last edge has to be in the tour, the tour uses either the left first-set edge or the door incident to $x$.
    By a similar argument, the tour uses either the right first-set edge or the right second-set edge of $x$.
    Lastly, the tour has a subtour in the XOR gadget assigned to~$x$, if the order of the XOR gadget is nonzero.
    Combined with \cref{obs:xor_external_edges}, all of the above implies the lemma.
\end{proof}

Suppose we have a parity gadget related to an $H$-vertex $x$.
We assume that the vertices $X$ and $X'$ of the parity gadget are incident to doors of $x$ (i.e., they are on the second-set path of $x$); the case when the vertices $Y$ and $Y'$ are incident to doors of $x$ can be argued analogously.
Let $\tau$ be a (not necessarily standard) tour of $G$.
We say that the parity gadget is \defi{engaged} at $x$ in $\tau$, if in the subtour of $\tau$ in the parity gadget, the degrees of $X$ and $X'$ are one.
We say that the parity gadget is \defi{disengaged} at $x$ in $\tau$, if these degrees are two.
In either case, we say that the parity gadget is \defi{strict} at $x$ in $\tau$. We will soon see that the strict gadget is always strict, motivating its name. 
Note that when the parity gadget is not strict at $x$ in $\tau$, it means that the degrees of $X$ and $X'$ in the subtour are not the same.
By \cref{lem:unique_subtours} and Figures~\ref{fig:strict_gadget_1} and \ref{fig:flexible_gadget}, we can characterize exactly when this happens.

\begin{observation}
\label{obs:strict_iff_standard}
	Let $\alpha$ be a parity gadget in $G$, $\tau$ be a tour of $G$, and $x$ be a related $H$-vertex of $\alpha$.
	Then $\alpha$ is strict at $x$ in $\tau$, if and only if the subtour of $\tau$ in $\alpha$ is standard.
\end{observation}

Note that this is the motivation for the names ``strict gadget'' and ``flexible gadget'': On the one hand, the strict gadget is always strict at a related $H$-vertex in every tour, since it does not have any non-standard subtours.
On the other hand, the flexible gadget has non-standard subtours, and hence there can be a tour $\tau$ in which the gadget is not strict at a related $H$-vertex.

Next, we have the following lemma about the edges in a standard tour of $G$.

\begin{lemma}
\label{lem:tour_edges}
    For any standard tour of $G$ and any $H$-vertex~$x$, exactly one of these cases holds:
    \begin{itemize}
        \item The first-set case: The tour uses the left and right first-set edges of~$x$ and does not use any doors or the right second-set edge of~$x$. 
        It also uses the subtour in the XOR gadget assigned to~$x$ incident to the left and right first-set edges.
        Further, all parity gadgets related to $x$ are disengaged at $x$.
        \item The second-set case: The tour uses all the doors and the right second-set edge of~$x$ and does not use the left and right first-set edges of~$x$. 
        It also uses the subtour in the XOR gadget assigned to~$x$ incident to the right second-set edge of $x$ and the closest door to~$x_r$.
        Further, all parity gadgets related to $x$ are engaged at $x$.
    \end{itemize}
\end{lemma}
\begin{proof}
    Recall that the second-set path of $x$ was $(x_{\ell}, x_{1}, x'_{1}, \dots, x_{d(x)}, x'_{d(x)}, x_r)$, where $d(x)$ denotes the degree of $x$ in $H$.
    For ease of argument, we define $x'_0 := x_{\ell}$, and we define $x_{d(x)+1}$ to be the endpoint of the closest door to $x_r$ other than $x'_{d(x)}$ (i.e., $x'_{d(x)}x_{d(x)+1}$ is an external edge of the XOR gadget assigned to~$x$). 
    
    We start by showing that either all the doors of $x$ are in the tour or all doors of $x$ are not in the tour.
    For the sake of contradiction, suppose that this does not hold.
    This implies that there are two edges $x'_{i-1} x_i$ and $x'_i x_{i+1}$ such that one of them is in the tour and the other one is not.
    This implies that the parity gadget that contains $x_i$ and $x'_i$ is not strict at $x$ in the tour.
    By \cref{obs:strict_iff_standard}, the subtour of the tour in this parity gadget is not standard, a contradiction to the assumption that the tour is standard.

    By definition of being engaged and disengaged, it is easy to see that when all doors are in the tour, all parity gadgets related to $x$ are engaged at $x$. 
    Otherwise, all these gadgets are disengaged at $x$.
    
    The lemma then follows from the two preceding paragraphs, \cref{obs:xor_external_edges}, \cref{lem:end_doors}, and the fact that the tour has to use a subtour in the XOR gadget assigned to~$x$.
\end{proof}

For an $H$-vertex $x$, we define an \defi{$x$-change} in a standard tour as the swap that can be performed to change from the first-set case of \cref{lem:tour_edges} to the second-set case, or vice versa. 
The following lemma shows that such a swap always exists and 
it specifies the number of edges to be swapped for an $x$-change.

\begin{lemma}
\label{lem:x-change}
    For an $H$-vertex $x$, an $x$-change is a swap of at least $k$ edges.
    Further, it is a swap of exactly $k$ edges, if and only if the $x$-change only involves standard subtour changes.
\end{lemma}
\begin{proof}
    We consider the $x$-change from the first-set case of \cref{lem:tour_edges} to the second-set case.
    For the other direction, we just switch the two sets of the swap, and hence the number of edges involved is the same.
    Recall that $d(x)$ is the degree of $x$ in $H$ (and in $\overrightarrow{H}$) and $d^+(x)$ is the out-degree of $x$ in $\overrightarrow{H}$.
    We count the involved edges in the following three parts of the $x$-change.
    
    Firstly, we need to remove the left first-set edge and add $d(x)+1$ doors.
    Hence, the total involved edges for this part is $d(x) + 2$.
    
    Secondly, if the order $p$ of the XOR gadget assigned to $x$ is nonzero, we need to remove the right first-set edge and add the right second-set edge.
    In addition, we need to change the subtour in the XOR gadget, by removing and adding $p-1$ edges.
    In any case, the number of involved edges is $2p$.
    
    Lastly, we need to make the change in all parity gadgets related to $x$; in other words, these gadgets change from being disengaged to being engaged at $x$.
    If we only use standard subtour changes, then by construction and \cref{ob:num_exchanged_edges}, in total, 
    we involve exactly $d^+(x) (2 \cdot 4 - 1) + \big(d(x) - d^+(x)\big) (2 \cdot 2 - 1) = 3d(x) + 4d^+(x)$.
    Otherwise, by the last condition in \cref{def:rx_ry_pg}, the number of removed and added internal edges would be more than $3d(x) + 4d^+(x)$.
    
    Adding up all of the above and noting that $p = k - 2d(x)-  2d^+(x) - 1$, we obtain that the total number of involved edges is exactly $d(x) + 2 + 2p + 3 d(x) +  4d^+(x) = 2k$ edges if we only use standard subtour changes, yielding a swap of exactly $k$ edges in this case.
    If there is some nonstandard subtour change, then the total number of involved edges is more than $2k$, and hence, it cannot be a swap of exactly $k$ edges.
    The lemma then follows.
\end{proof}

We define $\phi$ to be the function that maps a standard tour $\tau$ of $G$ to a cut $\gamma$ of $H$ such that for every $H$-vertex $x$, $x$ is in the first set, if and only if the first-set case of \cref{lem:tour_edges} holds for $G$ and $x$.

\begin{lemma}
\label{lem:phi_bijection}
	The function $\phi$ is a bijection between the standard tours in $G$ and the cuts in $H$.
	Further, for any cut $\gamma$ of $H$, $\phi^{-1}(\gamma)$ can be constructed in polynomial time.
\end{lemma}
\begin{proof}
	\cref{lem:tour_edges} implies that $\phi$ is injective.
	We now prove that it is surjective.
	Let $\gamma$ be an arbitrary cut of $H$.
	We construct a standard tour $\tau$ in $G$ as follows.
	After initializing the tour to be the empty graph, we add the incident edges of all degree-two vertices in $G$ to the tour.
	Then for each $H$-vertex $x$, if $x$ is in the first set of the cut, we add the left and right first-set edges of $x$ and the subtour in the XOR gadget assigned to $x$ incident to these edges.
    Otherwise, we add all the doors and the right second-set edge of $x$, as well as the subtour in the XOR gadget assigned to $x$ incident to the right second-set edge of $x$ and the closest door to $x_r$.
	Finally, for each $H$-edge $xy$, we add the following subtour in the corresponding gadget: We use the subtour (1) if $x$ and $y$ are in the first set, subtour (2) if $x$ is in the second set and $y$ in the first set, subtour (3) if $x$ is in the first set and $y$ in the second set, and subtour (4) if $x$ and $y$ are in the second set.
	By construction, it is easy to see that $\tau$ is a standard tour of $G$ and that $\phi(\tau) = \gamma$.
	The lemma then follows. 
\end{proof}

We now prove the correspondence between the flip sequence in $H$ and the $k$-swap sequence in $G$ in the following two lemmata.

\begin{lemma}
\label{lem:correspondence_13_1}
    Let $\tau_1$ be a standard tour of $G$, and let $\gamma_1 := \phi(\tau_1)$.
    Then for any cut $\gamma_2$ obtained from $\gamma_1$ by an improving flip, the tour $\phi^{-1}(\gamma_2)$ can be obtained from $\tau_1$ by an improving $k$-swap.
\end{lemma}
\begin{proof}
    Suppose the cuts $\gamma_1$ and $\gamma_2$ differ by an $H$-vertex $x$. 
    Let $\tau_2$ be the tour obtained by performing an $x$-change on $\tau_1$ with only standard subtour changes.
    By the definition of $x$-change, it is easy to see that $\tau_2 = \phi^{-1}(\gamma_2)$.
    Further, by \cref{lem:x-change}, we can transform $\tau_1$ to $\tau_2$ by performing a $k$-swap.

    Hence, what remains to be shown is that the $x$-change is improving.
    Suppose $P$ and $Q$ are the two sets of $\gamma_1$, such that $x \in P$. (We do not assume which one of $P$ and $Q$ is the first set.) 
    Then for the change from $\gamma_1$ to $\gamma_2$, the change in value for the Max-Cut instance is
    \[
      - \sum_{xy \in E(H) : y \in Q} w(xy) + \sum_{xy \in E(H) : y \in P} w(xy) := \Delta.
    \]
    Since $\gamma_2$ is the result of an improving flip from $\gamma_1$, we have $\Delta > 0$.

    By \cref{lem:unique_subtours}, for every parity gadget, the total weights of subtours (1) and (4) are the same-set weight of the gadget, while those of subtours (2) and (3) are the different-set weight.
    Further, note that for each parity gadget related to $x$, by the definition of the $x$-change, the subtour in the gadget in exactly one of $\tau_1$ and $\tau_2$ is either subtour (1) or (4) and that in the other tour is either subtour (2) or (3).    
    Since $\gamma_1 = \phi(\tau_1)$ and $\gamma_2 = \phi(\tau_2)$, by the preceding two sentences and the assignment of the same-set and different-set weights to the gadgets, the change of the total edge weight of the tour from $\tau_1$ to $\tau_2$ is
    \begin{align*}
        & - \Big(\sum_{\substack{xy \in E(H) : \\ y \in P, w(xy) \geq 0}} w(xy) - \sum_{\substack{xy \in E(H) : \\ y \in Q, w(xy) < 0}} w(xy) \Big) + \Big(\sum_{\substack{xy \in E(H) : \\ y \in Q, w(xy) \geq 0}} w(xy) - \sum_{\substack{xy \in E(H) : \\ y \in P, w(xy) < 0}} w(xy) \Big) \\
        = & \sum_{\substack{xy \in E(H) : \\ y \in Q}} w(xy) - \sum_{\substack{xy \in E(H) : \\ y \in P}} w(xy) = -\Delta < 0.
    \end{align*}

    The lemma then follows.    
\end{proof}

The next lemma shows the reverse correspondence.
\begin{lemma}
\label{lem:correspondence_13_2}
    Let $\tau_1$ be a standard tour of $G$, and let $\gamma_1 = \phi(\tau_1)$.
    Then for any standard tour $\tau_2$, obtained from $\tau_1$ by performing an improving $k$-swap, the cut $\phi(\tau_2)$ can be obtained from $\gamma_1$ by an improving flip.
\end{lemma}
\begin{proof}
    If $\tau_2$ is obtained from $\tau_1$ by performing an $x$-change for some $H$-vertex $x$, then we can define~$\gamma_2$ to be the cut obtained from~$\gamma_1$ by flipping $x$.
    Using an analogous argument as in the proof of \cref{lem:correspondence_13_1}, we can see that $\gamma_2 = \phi(\tau_2)$ and $\gamma_2$ can be obtained from $\gamma_1$ by an improving flip.
    Hence, it remains to prove that any $k$-swap from $\tau_1$ to another standard tour has to be an $x$-change, for some $x$.

    By \cref{lem:tour_edges}, if we remove a first-set edge of an $H$-vertex~$x$, a door of $x$, or the second-set edge of~$x$, then we have to do an $x$-change in order to maintain a standard tour.
    Since the $x$-change already involves a swap of at least $k$ edges, we cannot swap any other edge.
    If we do not change any first-set edge or edges on a second-set path, then by the design of parity gadgets, we cannot change the subtour in any gadget.
    As a result, the tour remains the same.
    This completes the proof of the lemma.
\end{proof}

Lemmas~\ref{lem:phi_bijection}, \ref{lem:correspondence_13_1}, and \ref{lem:correspondence_13_2} immediately imply the following.

\begin{corollary}
\label{lem:correspondence}
	The transition graph of the Max-Cut instance $(H,w)$ is isomorphic to the subgraph of the transition graph of the TSP/\kopt instance $(G_{\infty}, c_{\infty})$ induced by the standard tours of $G$.
\end{corollary}

\section{Non-standard tours}
\label{sec:flexible}

In \cref{sec:correspondence}, we show a one-to-one correspondence between the standard tours of $G$ and the cuts of $H$ that preserves the local neighborhood.
We now argue that our construction ensures that $G$ does not have any non-standard tour.
The key idea is the following two forcing rules of strictness.

\begin{observation}
\label{lem:forcing_a}
    In a tour of $G$, if a parity gadget is strict at one related $H$-vertex, then it is also strict at the other related $H$-vertex. 
\end{observation}

\begin{lemma}
\label{lem:forcing_b}
    Suppose $x$ is an $H$-vertex with degree $d(x)$.
    If $d(x)-1$ parity gadgets related to $x$ are strict at $x$ in a tour of $G$, then the remaining parity gadget related to $x$ is also strict at $x$.
\end{lemma}
\begin{proof}
    By \cref{lem:end_doors}, the door incident to $x_{\ell}$ and the closest door to $x_r$ have to be either both present or both absent in the tour.
    Together with the strictness of the other $d(x)-1$ gadgets, it follows that either all doors of $x$ are in the tour or all doors of $x$ are not in the tour.
    This implies the strictness at $x$ for the remaining gadget.
\end{proof}

From the forcing rules above, we have a sufficient condition of the assignment of strict and flexible gadgets to ensure that all parity gadgets are always strict.

\begin{lemma}
\label{lem:strict_flexible_gadgets}
    If the $H$-edges corresponding to the flexible gadgets form a forest in $H$, then all parity gadgets are strict at their related $H$-vertices in any tour of $G$.
\end{lemma}
\begin{proof}
    Note that the strict gadget has only standard subtours, and hence, we only need to consider the flexible gadget.
    The lemma is obtained from a simple recursive argument on a leaf of the forest, as follows.
    Consider the subgraph of $H$ containing all $H$-edges corresponding to the flexible gadgets that have not been proved to satisfy the property in the lemma.
    Then this subgraph is a forest.
    Applying \cref{lem:forcing_b} to a leaf of the forest, we obtain the strictness at the leaf $H$-vertex for the flexible gadget corresponding to the edge incident to this leaf.
    Combined with \cref{lem:forcing_a}, this flexible gadget is strict at both its related $H$-vertices.
\end{proof}

We now show that our construction satisfies the condition in \cref{lem:strict_flexible_gadgets} and hence $G$ cannot have non-standard tours.

\begin{corollary}
\label{cor:no_nonstandard}
	All tours of $G$ are standard tours.
\end{corollary}
\begin{proof}
	By construction and \cref{lem:partialedgeorientation}, the $H$-edges corresponding to the flexible gadgets form a disjoint union of trees.
 Hence, by \cref{lem:strict_flexible_gadgets}, in any tour of $G$, all parity gadgets are strict at their related $H$-vertices.
	Combined with \cref{obs:strict_iff_standard}, this implies the corollary.
\end{proof}

\section{No local optimum with non-edges}
\label{sec:no_infty_opt}

Recall that the non-edges are the edges not in $G$ and are only added to make a complete graph $G_{\infty}$.
In this section, we describe the weight assignment to the non-edges, completing the weight function $c_{\infty}$ for this complete graph.
Moreover, we show that no locally optimal tour contains a non-edge (\cref{lem:no_infty_local_optima} below).
Note that this is precisely where the substantial gap lies in Krentel's proof~\cite{Kre1989}. 
In particular, he assumes without proof that no non-edges can appear in a local optimum. 
The definition of \pls-completeness~(see \cref{def:pls_complete}) requires that the function $g$ maps local optima to local optima. 
Therefore, one has to show either that a local optimum cannot contain non-edges or how to extend the definition of the function $g$ to local optima that contain non-edges. 
Neither are done in the paper of Krentel~\cite{Kre1989}, and there is no obvious way how to fill this gap. 
We are not aware of a generic way to prove such a result for arbitrary TSP instances as for example those constructed by Krentel~\cite{Kre1989}. For proving the PLS-completeness of the Lin-Kernighan heuristic Papadimitriou~\cite{pap1992}  showed that non-edges cannot be part of a local optimum. However, his argument only applies to the special construction he used for his proof. 

We start with the following observation, based on the construction of $G$.

\begin{observation}
\label{lem:degrees}
    The following holds for the graph $G$:
    \begin{enumerate}[(a)]
        \item \label{item:all_degrees} Each vertex of $G$ has degree two, three, or four.
        \item \label{item:adjacent_degree_two} Every vertex in $G$ is either a degree-two vertex or adjacent to one.
        \item \label{item:degree_four} Every degree-four vertex in $G$ is in a parity gadget.
        Every strict gadget in $G$ has four degree-four vertices (namely, $Z, Z', a$, and $Y'$). 
        Every flexible gadget in $G$ also has four degree-four vertices (namely, $Z, Z', X'$, and $Y'$). 
        \end{enumerate}
\end{observation}

Let $V_2$, $V_3$, and $V_4$ be the vertices in $G$ that have degree two, three, and four, respectively.
Further, denote $N_2 := |V_2|$, $N_3 := |V_3|$, $N_4 := |V_4|$, and $N = N_2 + N_3 + N_4$.
By \cref{lem:degrees}(\ref{item:all_degrees}), $(V_2, V_3, V_4)$ is a partition of the vertex set of $G$, and $N$ is the total number of vertices in $G$.
In this section, when we say degree-two, -three, -four vertices, the degree we refer to is the degree in $G$.

\paragraph{Weight assignment of non-edges.}
We first assign a \defi{priority} $\lambda(v)$ to every vertex~$v$ of $G$ in the decreasing order from $N$ to $1$, as follows.
Firstly, we assign the values in $\{N, \dots, N_3 + N_4 + 1\}$ to the vertices in $V_2$ in an arbitrary order.
Secondly, we assign the values in $\{N_3 + N_4, \dots, N_3 + N_4 - n + 1\}$ to the vertices of the form $x_{\ell}$ in an arbitrary order.
(Recall that $x_{\ell}$ and $x_r$ are the endpoints of the original first-set edge of an $H$-vertex~$x$.)
Thirdly, for each $H$-vertex $x$ (in some arbitrary order), we assign the next available priority values to $x_r$ and the vertices in the XOR-gadget assigned to $x$ in the order of $(a_1, \dots, a_p, x_r, b_p, \dots, b_1)$, where $(a_1, \dots, a_p)$ and $(b_1, \dots, b_p)$ are the two rails of the XOR-gadget assigned to $x$, and $a_1$ and $b_p$ are incident to $x_{\ell}$ and $x_r$, respectively.
Lastly, for each parity gadget~$\alpha$ in the order $\psi$ as defined in \cref{sec:equip}, we assign the next available values of priority to all vertices in~$\alpha$ in the following order: $(Y, X, b, X', b', a', Z, Z', a, Y')$ if $\alpha$ is the strict gadget and $(Y, X, Z, Z', X', Y')$ if $\alpha$ is the flexible gadget.
Note that we assign priorities to the degree-three vertices before assigning them to the degree-four vertices in $\alpha$.
Next, let $M$ be the sum of all edge weights appearing in $G$.
For a non-edge $vv'$, we assign $M \cdot 4^{\max\{\lambda(v), \lambda(v')\}}$ to be the weight of $vv'$.
This completes the definition of the TSP instance $(G_{\infty}, c_{\infty})$. Immediately from this definition we get:
\begin{observation}
    Let $(E_1, E_2)$ be a 3-swap in $(G_{\infty}, c_{\infty})$ such that $E_1$ contains a \infedge
    $e_1$ with $c_{\infty}(e_1) > \max_{e\in E_2} c_{\infty}(e)$.
    Suppose that performing the swap $(E_1, E_2)$ from a tour $\tau$ of $(G_{\infty}, c_{\infty})$ results in another tour.
    Then this swap is improving.
    \label{obs:improving3swap}
\end{observation}

The following lemma essentially shows that a tour with a non-edge cannot be locally optimal.

\begin{lemma}
    \label{lem:no_infty_local_optima}
    Let $\tau$ be a tour of $(G_{\infty}, c_{\infty})$ containing a \infedge.
    Then there exists an improving $3$-swap for $\tau$.
\end{lemma}
\begin{proof}
	We call an edge of $G$ a \defi{\fedge}.
	For convenience, we still use the word ``edge" to refer generally to either a \infedge or a \fedge.

    We define the cut $\gamma^*$ of $H$ as follows.
    For every $H$-vertex $x$, if the left first-set edge of $x$ is in $\tau$, then $x$ is in the first set of $\gamma^*$; otherwise it is in the second set.
    Let $\tau^* := \phi^{-1}(\gamma^*)$.
	    
    We say that a vertex of $G_{\infty}$ is \defi{$\tau^*$-consistent}, if the incident edges of that vertex in~$\tau$ are the same as those in $\tau^*$.
    Let $i$ be the smallest integer such that for any vertex $v$ of $G_{\infty}$ with $\lambda(v) > i$, $v$ is $\tau^*$-consistent.
    Since $\tau$ has a \infedge, $i \geq 1$.
    Let $v$ be the vertex with $\lambda(v) = i$.
    By definition of $i$,  $v$ is $\tau^*$-inconsistent.

    We start with the following claim.

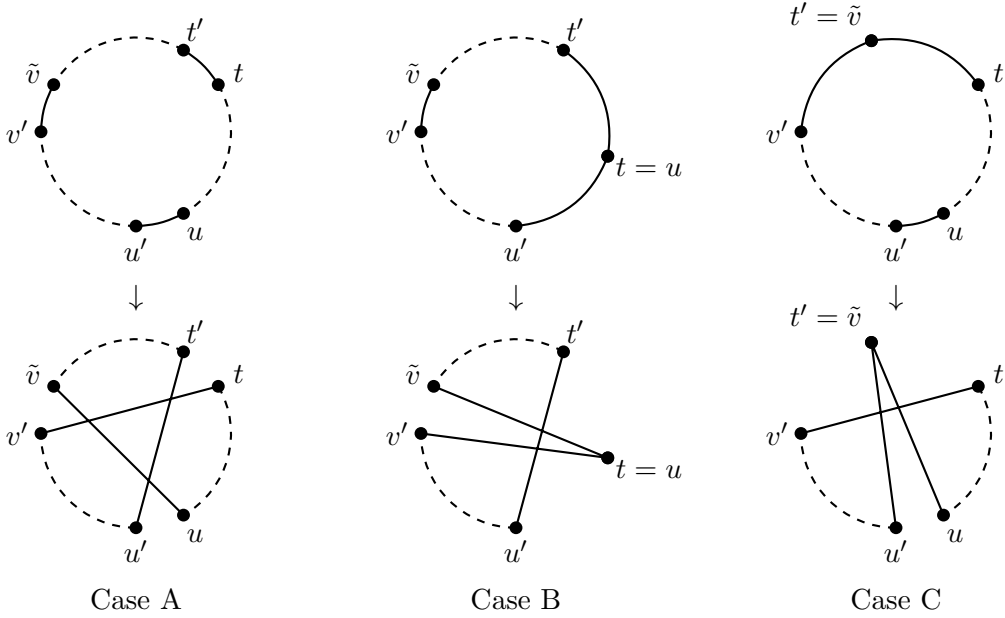
\begin{figure}
    \centering
    \begin{tikzpicture}[vertex/.style = {circle,fill=black,inner sep=0.6mm},
                        dashededge/.style={dashed, thick},     
                        solidedge/.style={thick}]     
        \def\radius{1.25} 

        \begin{scope}
            \foreach \i/\j/\t in {1/1/$t$,  2/2/$t'$, 3/5/$\tilde{v}$, 4/6/$v'$, 5/9/$u'$, 6/10/$u$} {
                \coordinate (P\i) at ({\radius*cos(30*\j)}, {\radius*sin(30*\j)});
                \node[vertex] at (P\i) {};
                \node at ({1.25*\radius*cos(30*\j)}, {1.25*\radius*sin(30*\j)}) {\t};
            }
            \draw[solidedge] (P1) to[bend right=12.5] (P2) (P3) to[bend right=12.5] (P4) (P5) to[bend right=12.5] (P6);
            \draw[dashededge] (30:\radius) arc[start angle=30, end angle=-60, radius=\radius];
            \draw[dashededge] (-90:\radius) arc[start angle=-90, end angle=-180, radius=\radius];
            \draw[dashededge] (60:\radius) arc[start angle=60, end angle=150, radius=\radius];
            \draw (0,-2.2) node {$\downarrow$};       
        \end{scope}
        
        \begin{scope}[shift={(0,-4)}]
            \foreach \i/\j/\t in {1/1/$t$,  2/2/$t'$, 3/5/$\tilde{v}$, 4/6/$v'$, 5/9/$u'$, 6/10/$u$} {
                \coordinate (P\i) at ({\radius*cos(30*\j)}, {\radius*sin(30*\j)});
                \node[vertex] at (P\i) {};
                \node at ({1.25*\radius*cos(30*\j)}, {1.25*\radius*sin(30*\j)}) {\t};
            }
            \draw[solidedge] (P1) -- (P4) (P2) -- (P5) (P3) -- (P6);
            \draw[dashededge] (30:\radius) arc[start angle=30, end angle=-60, radius=\radius];
            \draw[dashededge] (-90:\radius) arc[start angle=-90, end angle=-180, radius=\radius];
            \draw[dashededge] (60:\radius) arc[start angle=60, end angle=150, radius=\radius];
            \draw (0,-2.2) node {Case A};
        \end{scope}

        \begin{scope}[shift={(5,0)}]
            \foreach \i/\j/\t in {1/-0.5/,  2/2/$t'$, 3/5/$\tilde{v}$, 4/6/$v'$, 5/9/$u'$, 6/-0.5/} {
                \coordinate (P\i) at ({\radius*cos(30*\j)}, {\radius*sin(30*\j)});
                \node[vertex] at (P\i) {};
                \node at ({1.25*\radius*cos(30*\j)}, {1.25*\radius*sin(30*\j)}) {\t};
            }
            \node at ({1.45*\radius*cos(-15)}, {1.45*\radius*sin(-15)}) {$t = u$};
            \draw[solidedge] (P1) to[bend right=31.25] (P2) (P3) to[bend right=12.5] (P4) (P5) to[bend right=31.25] (P6);
            \draw[dashededge] (-90:\radius) arc[start angle=-90, end angle=-180, radius=\radius];
            \draw[dashededge] (60:\radius) arc[start angle=60, end angle=150, radius=\radius];
            \draw (0,-2.2) node {$\downarrow$};       
        \end{scope}

        \begin{scope}[shift={(5,-4)}]
            \foreach \i/\j/\t in {1/-0.5/,  2/2/$t'$, 3/5/$\tilde{v}$, 4/6/$v'$, 5/9/$u'$, 6/-0.5/} {
                \coordinate (P\i) at ({\radius*cos(30*\j)}, {\radius*sin(30*\j)});
                \node[vertex] at (P\i) {};
                \node at ({1.25*\radius*cos(30*\j)}, {1.25*\radius*sin(30*\j)}) {\t};
            }
            \node at ({1.45*\radius*cos(-15)}, {1.45*\radius*sin(-15)}) {$t = u$};
            \draw[solidedge] (P1) -- (P4) (P2) -- (P5) (P3) -- (P6);
            \draw[dashededge] (-90:\radius) arc[start angle=-90, end angle=-180, radius=\radius];
            \draw[dashededge] (60:\radius) arc[start angle=60, end angle=150, radius=\radius];
            \draw (0,-2.2) node {Case B};       
        \end{scope}

        \begin{scope}[shift={(10,0)}]
            \foreach \i/\j/\t in {1/1/$t$,  2/3.5/, 3/3.5/, 4/6/$v'$, 5/9/$u'$, 6/10/$u$} {
                \coordinate (P\i) at ({\radius*cos(30*\j)}, {\radius*sin(30*\j)});
                \node[vertex] at (P\i) {};
                \node at ({1.25*\radius*cos(30*\j)}, {1.25*\radius*sin(30*\j)}) {\t};
            }
            \node at ({1.3*\radius*cos(105)-0.5}, {1.3*\radius*sin(105)}) {$t' = \tilde{v}$};
            \draw[solidedge] (P1) to[bend right=31.25] (P2) (P3) to[bend right=31.25] (P4) (P5) to[bend right=12.5] (P6);
            \draw[dashededge] (30:\radius) arc[start angle=30, end angle=-60, radius=\radius];
            \draw[dashededge] (-90:\radius) arc[start angle=-90, end angle=-180, radius=\radius];
            \draw (0,-2.2) node {$\downarrow$};       
        \end{scope}

        \begin{scope}[shift={(10,-4)}]
            \foreach \i/\j/\t in {1/1/$t$,  2/3.5/, 3/3.5/, 4/6/$v'$, 5/9/$u'$, 6/10/$u$} {
                \coordinate (P\i) at ({\radius*cos(30*\j)}, {\radius*sin(30*\j)});
                \node[vertex] at (P\i) {};
                \node at ({1.25*\radius*cos(30*\j)}, {1.25*\radius*sin(30*\j)}) {\t};
            }
            \node at ({1.3*\radius*cos(105)-0.5}, {1.3*\radius*sin(105)}) {$t' = \tilde{v}$};
            \draw[solidedge] (P1) -- (P4) (P2) -- (P5) (P3) -- (P6);
            \draw[dashededge] (30:\radius) arc[start angle=30, end angle=-60, radius=\radius];
            \draw[dashededge] (-90:\radius) arc[start angle=-90, end angle=-180, radius=\radius];
            \draw (0,-2.2) node {Case C};       
        \end{scope}        
    \end{tikzpicture}
    \caption{Illustrations for the cases in the proof of \cref{claim:3_swap_from_infty}. Note that $u\tilde{v}$ by assumption is a $G$-edge.}
    \label{fig:3_swap_from_infty}
\end{figure}

    \begin{claim}
    \label{claim:3_swap_from_infty}
        Suppose there exist a \infedge $\tilde{v}v'$ in~$\tau$ and an edge $u\tilde{v}$ in~$\tau^*$ but not in~$\tau$.
        Further suppose that for any vertex~$v'' \neq u,\tilde{v}$, if $v''$ is not $\tau^*$-consistent, then $\lambda(v'') < \lambda(\tilde{v})$.
        Then there exists an improving 3-swap that removes $\tilde{v}v'$.
    \end{claim}
    \begin{claimproof}[Claim's proof]
        Fix an orientation of $\tau$.
        Let $u'$ be a vertex such that $uu'$ is in $\tau$ but not in $\tau^*$.
        This implies that $u'$ is not $\tau^*$-consistent, and hence, $\lambda(u') < \lambda(\tilde{v})$.
        Similarly, $\lambda(v') < \lambda(\tilde{v})$.
        If $(\tilde{v},v')$ and $(u,u')$ have the same orientation in $\tau$, then since $u\tilde{v}$ is a \fedge, and since $\lambda(u')$ and $\lambda(v')$ are smaller than $\lambda(\tilde{v})$, replacing $\tilde{v}v'$ and $uu'$ by $u\tilde{v}$ and $u'v'$ yields a better tour by \cref{obs:improving3swap}.
        Otherwise, consider the path in $\tau$ between $u$ and $\tilde{v}$ without containing $u'$ and $v'$.
        If this path has only edges in $\tau^*$, then the path together with $u\tilde{v}$ forms a cycle in $\tau^*$ that is not $\tau^*$ itself (since $u'$ and $v'$ are not on the path), a contradiction to $\tau^*$ being a tour.
        Therefore, the path contains an edge $tt'$ that is not in $\tau^*$, where we encounter $t$ before $t'$ when we go along the path from $u$ to $\tilde{v}$.
        This implies that $t$ and $t'$ are not $\tau^*$-consistent.
        
        Let $E_1 = \{ uu', \tilde{v}v', tt' \}$ and $E_2 = \{ u\tilde{v}, t'u', tv' \}$.
        We now show that the swap $(E_1, E_2)$ is improving, by analyzing the following cases; see \cref{fig:3_swap_from_infty} for an illustration.

        \textbf{Case A: $t \neq u$ and $t' \neq \tilde{v}$.}
        This implies that $\lambda(t)$ and $\lambda(t')$ are smaller than $\lambda(\tilde{v})$. 
        By \cref{obs:improving3swap}, replacing $E_1$ with $E_2$ yields a shorter tour.

        \textbf{Case B: $t = u$.}
        Observe that if the vertex $v$ (i.e., the $\tau^*$-inconsistent vertex with the highest priority in $\tau$) has degree at least three, then by the priority assignment, all the degree-two vertices are $\tau^*$-consistent.
        Combined with \cref{lem:degrees}(\ref{item:adjacent_degree_two}), this implies that in $\tau$, every vertex is incident to an edge that also appears in $\tau^*$.
        It follows that $t \neq u$ and $t' \neq \tilde{v}$, a contradiction to the assumption of this case.

        Hence, $v$ is a degree-two vertex.
        Note that by the assumption of the claim, one of $\tilde{v}$ and $u$ is $v$.
        For this vertex, any incident edge that does not appear in $\tau^*$ has to be a \infedge and has weight $M \cdot 4^i$, i.e., the maximum weight among all edges of $\tau$.
        If $\tilde{v} = v$, then the edge $\tilde{v}v'$ has larger weight than the weight of every edge in $E_2$.
        By \cref{obs:improving3swap}, the swap $(E_1, E_2)$ is then improving.
        Otherwise, we have $u = v$.
        Then both $uu'$ and $tt'$ have weight $M \cdot 4^i$.
        Since 
        $c_{\infty}(tv') \leq M \cdot 4^i$ and $c_{\infty}(u\tilde{v}), c_{\infty}(t'u') \leq M \cdot 4^{i-1}$, replacing $E_1$ with $E_2$ yields a shorter tour.
        
        \textbf{Case C: $t' = \tilde{v}$.}
        Using a similar argument as in Case B, we also conclude that $v$ has degree two.
        If $\tilde{v} = v$, then both $\tilde{v}v'$ and $tt'$ have weight $M \cdot 4^i$.
        Since $c_{\infty}(t'u') \leq M \cdot 4^i$ and $c_{\infty}(u\tilde{v}), c_{\infty}(tv') \leq M \cdot 4^{i-1}$, replacing $E_1$ with $E_2$ yields a shorter tour.
        Otherwise, we have $u = v$.
        Then the edge $uu'$ has larger weight than the weight of every edge in $E_2$, and hence the swap $(E_1, E_2)$ is improving by \cref{obs:improving3swap}.
    \end{claimproof}

    By \cref{claim:3_swap_from_infty}, if in $\tau$ the vertex $v$ is incident to a \infedge, we have an improving 3-swap for $\tau$.
    Therefore, for the remainder of the proof, we assume that in $\tau$ the vertex $v$ is incident to only \fedges.
    We consider the following five cases.
    
    \textbf{Case 1: $v \in V_{2}$.} 
    In this case, $v$ only has two incident \fedges, and they have to be both in $\tau^*$ and in $\tau$.
    Therefore, $v$ is $\tau^*$-consistent, a contradiction.

    \textbf{Case 2: $v$ has degree~$d\in\{3,4\}$ in $G$ and has at least $d-1$ neighbors in $G$ with priority higher than $i$.}
    By the choice of $i$, the neighbors with priority higher than $i$ are $\tau^*$-consistent.
    Hence, the edges that are incident to $v$ and one of these neighbors are either both in $\tau$ and $\tau^*$ or both not in $\tau$ and $\tau^*$.
    Hence, if the two incident edges of $v$ in $\tau^*$ are incident to two of these neighbors, then $v$ is $\tau^*$-consistent, a contradiction.
    Otherwise, $v$ is adjacent to exactly one of these neighbors in both $\tau$ and $\tau^*$.
    This implies the edge $vv'$ is in both $\tau$ and $\tau^*$, where $v'$ is the remaining neighbor of $v$ in $G$.
    However, this also means that $v$ is $\tau^*$-consistent, a contradiction.

    \textbf{Case 3: $v = x^j_{\ell}$ for some $H$-vertex~$x$ and $j \in \{1, \dots, n\}$.}
    Observe that $v$ is incident to the left first-set edge of~$x$, a door $vv'$ of $x$, and an edge $uv$, where $u \in V_{2}$.
    By the choice of $i$, $u$ is $\tau^*$-consistent, and hence, $uv$ is both in $\tau$ and $\tau^*$.
    By construction of $\tau^*$, if the left first-set edge of $x$ is in $\tau$, then it is also in $\tau^*$.
    Otherwise, $vv'$ is in $\tau^*$, and by the assumption that $v$ is incident to only \fedges, $vv'$ is also in $\tau$.
    Hence, $v$ is $\tau^*$-consistent, a contradiction.

    \textbf{Case 4: $v$ is in $V_3$ and not covered by Cases 2 and 3.}
    By \cref{lem:degrees}, the degree-three vertices are $x_{\ell}, x_r$ for $x \in V(H)$ and some vertices in the gadgets.
    Note that $x_{\ell}$ is covered by Case 3, while by construction and the definition of $\lambda$, $x_r$ and the degree-three vertices in XOR gadgets are covered by Case 2.
    The degree-three vertices in a flexible gadget are the vertices $X$ and $Y$ of the gadget, and by the priority assignment using the order $\psi$, they are covered by Case 2.
    The same holds for the degree-three vertices $Y$, $X$, $b$, $b'$, $a'$ in a strict gadget.
    Hence, $v$ can only be the vertex $X'$ of some strict gadget~$\alpha$.
    
    In this case, the two neighbors of $v$ that have degree more than two in $\alpha$ are the vertex $b'$ of $\alpha$ and a vertex $v'$ outside of the gadget~$\alpha$.
    We consider two subcases.
    (See \cref{fig:case_3} for an illustration.)

    \begin{figure}[t!]
        \centering
        \includegraphics[scale=1]{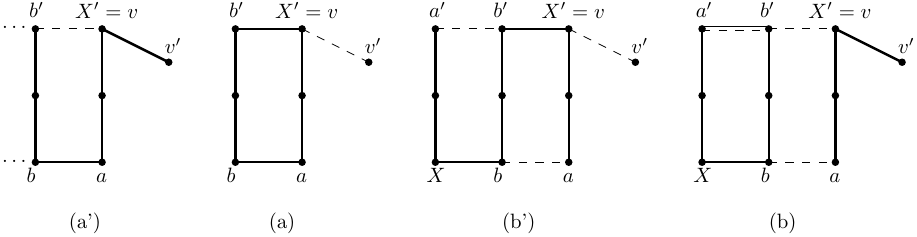}
        \caption{Illustration of Subcases 4a and 4b in the proof of \cref{lem:no_infty_local_optima}. 
        Bold edges are edges of $\alpha$ present in a specific tour, while dashed edges are not present.
        Panels (a') and (b') corresponds to the tour $\tau^*$ in Subcases 4a and 4b, respectively.
        Since $X, b$, and $a$ are $\tau^*$-consistent, these panels imply the situations in the tour $\tau$ in panels (a) and (b).
        The edge $a'b'$ in panel (b) indicates that this edge may be in $\tau$ (which creates a 6-cycle) or not in $\tau$ (which means $b'$ is incident to a \infedge).}
        \label{fig:case_3}
    \end{figure}

    \textit{Subcase 4a: $vv'$ is in $\tau^*$ and $vb'$ is in $\tau$.}
    In this case, $ab$ has to be in $\tau^*$, by the design of the XOR gadget.
    By the priority assignment, $b$ has higher priority than $v$, and hence, $b$ is $\tau^*$-consistent.
    This implies that $ab$ is also in $\tau$.
    However, we then have a subgraph of $\tau$ that is the 6-cycle containing $a, b, b'$, and $v$, a contradiction to $\tau$ being a tour.
    
    \textit{Subcase 4b: $vv'$ is in $\tau$ and $vb'$ is in $\tau^*$.}
    In this case, $Xb$ has to be in $\tau^*$, by the design of the XOR gadget.
    With an analogous argument as before, this edge is also in $\tau$, and hence $a'b'$ must not be in $\tau$, because otherwise, we have a 6-cycle.
    Since $b'$ is not adjacent to $v$ or $a'$ in $\tau$, it must be incident to a \infedge $b'u'$.
    Note that $\lambda(b') = \lambda(v) -1 = i-1$.
    Then applying \cref{claim:3_swap_from_infty} to the edges $b'u'$ and $b'v$, we obtain an improving 3-swap that removes $b'u'$.
    
    \textbf{Case 5: $v$ is in $V_4$ and is not covered by Case 2.}
    By \cref{lem:degrees}(\ref{item:degree_four}), $v$ is in some parity gadget $\alpha$ (which can be either a strict or flexible gadget).
    By \cref{lem:degrees}(\ref{item:adjacent_degree_two}), in $\tau^*$, $v$ is adjacent to a vertex $u \in V_2$.
    Using a similar argument as in Case~4, we deduce that $uv$ is also in $\tau$, and $v$ has a neighbor~$v'$ in $\tau^*$, such that $vv'$ is not in $\tau$.
    Let $t$ be the vertex such that $vt$ is the edge in $\tau$ but not in $\tau^*$.
    Then we must have the following.
    \begin{itemize}
        \item[(*)] $v'$ and $t$ have lower priorities than $i$; and
        \item[(**)] A neighbor of $t$ in $\tau^*$ has lower priority than $i$.
    \end{itemize}
    Indeed, (*) follows from the fact that $v'$ and $t$ are $\tau^*$-inconsistent.
    To see (**), observe that if this does not hold, then all incident edges of $t$ in $\tau^*$ would also have been in $\tau$, contradicting the fact that $t$ is not $\tau^*$-consistent.
    Note that since $\tau^*$ is a standard tour, the subtour~$T$ of $\tau^*$ in $\alpha$ is among subtours (1)--(4) of $\alpha$.
    With all the observations above in mind, we now consider two subcases.

    \textit{Subcase 5a: $\alpha$ is a strict gadget.}
    Recall that the strict gadget and its standard subtours are depicted in \cref{fig:strict_gadget_1}.
    Since $v$ is not covered by Case 2, $v$ must then be the vertex $Z$ of $\alpha$.
    If $T$ is subtour (1) or (3), then there is no choice for $v'$, as all neighbors of $v$ in $T$ have higher priorities than $i$, contradicting (*).
    If $T$ is subtour (2), $t$ must be $Z'$; but the neighbors of $t$ in $T$ then have higher priorities than $i$, contradicting (**).
    Hence, $T$ is subtour (4).
    This implies $v'$ is $Z'$, and $t$ is $Y'$.
    Observe that in this subtour, the neighbors of $a$ have priority higher than $v$.
    Hence, the incident edges of $a$ in $\tau$ and $\tau^*$ are the same.
    Accordingly, the edge $Z'a$ cannot be in $\tau$.
    Since the edge $ZZ'$ is not in $\tau$ either, and since the other neighbors of $Z'$ in $\alpha$ other than $Z$ and $a$ have higher priority than $i$, we conclude that only one \fedge is incident to $Z'$ in $\tau$.
    This means $Z'$ is incident to a \infedge~$Z't'$ in $\tau$.
    Since $\lambda(Z') = \lambda(Z) - 1 = i-1$, and since $t'$ must not be $\tau^*$-consistent, the edges $Z't'$ and $Z'Z$ satisfy the condition of \cref{claim:3_swap_from_infty}.
    Hence, there is a 3-swap that removes $Z't'$.

    \textit{Subcase 5b: $\alpha$ is a flexible gadget.}
    Recall that the flexible gadget and its standard subtours are depicted in \cref{fig:flexible_gadget}.
    Since $v$ is not covered by Case 2, $v$ must then be the vertex $Z$ or $X'$ of $\alpha$.
    We start with the following observation which follows from the gadget assignment according to the order $\psi$, the priority assignment, and the fact that the paths in the subgraph of $\overrightarrow{H}$ induced by the undirected edges have length at most two (by \cref{lem:partialedgeorientation}).

    \begin{itemize}
        \item[(***)] For every flexible gadget assigned in $G$, at least $X'$ or $Y'$ is adjacent to a vertex that does not belong to any parity gadget (i.e., a vertex in an XOR-gadget or a vertex of the form $x_r$).
    \end{itemize}

    We now consider the different cases for $T$.
    Firstly, if $T$ is subtour (1), then $v$ cannot be $Z$, since otherwise, there is no valid choice for $v'$ that satisfies (*).
    In other words, $v$ is $X'$.
    We then have that $v'$ is $Y'$.
    Since $t$ is a neighbor of $X'$ with priority lower than $i$, $t$ must then be a vertex of another parity gadget.
    Then by (***), in $G$, $v'$ (i.e., vertex $Y'$ of $\alpha$) is adjacent to a vertex $v''$ not in any parity gadget.
    By the priority assignment and by the fact that $v$ is in a parity gadget, this implies that $v''$ has priority higher than $i$.
    Therefore, since the edge $v'v''$ is not in $\tau^*$, it is not in $\tau$.
    By a similar reason, $Zv'$ cannot be in $\tau$.
    As $vv'$ is not in $\tau$ either, we conclude that $v'$ is incident to a \infedge $v't'$ in $\tau$.
    Note that $\lambda(v') = \lambda(v) - 1 = i-1$, and $t'$ must be $\tau^*$-inconsistent and have priority lower than $i$.
    Hence, $vv'$ and $v't'$ satisfy the condition of \cref{claim:3_swap_from_infty}, and therefore, there is an improving 3-swap.

    Secondly, 
    if $T$ is subtour (2), then $v$ must be $Z$ or $X'$.
    If $v$ is $Z$, then $t$ is $Z'$, and if $v$ is $X'$, then $t$ is $Y'$.
    In both cases, all neighbors of $t$ in $T$ also have higher priorities than that of $v$, a contradiction to (**).

    Thirdly, if $T$ is subtour (3).
    Regardless of $v$ is $Z$ or $X'$, we have no valid choice for $v'$.

    Lastly, suppose $T$ is subtour (4).
    Then by (***), at least one of $X'$ and $Y'$ is adjacent to a vertex $v''$ not in any parity gadget and therefore $v''$ has priority higher than $i$.
    If $v''$ is the neighbor of $X'$, then $v$ must be $Z$, since there is no valid choice for $v'$ when $v$ is $X'$.
    Further, the neighbors of $X'$ in $\tau^*$ have priorities higher than $i$, and hence the incident edges of $X'$ in $\tau$ and $\tau^*$ are the same.
    With similar argument as before, we then conclude that $Z'$ cannot be adjacent to $X'$, $Z$, or $Y$ in $\tau$, and hence it has to be incident to a \infedge $Z't'$.
    We then can apply \cref{claim:3_swap_from_infty} on $ZZ'$ and $Z't'$ and find an improving 3-swap.
    
    Now suppose $v''$ is the neighbor of $Y'$.
    Regardless of $v$ is $Z$ or $X'$, $t$ is then $Y'$.
    However, for $T$ being subtour (4), the two neighbors of $t$ in $\tau^*$ have priorities higher than $i$, contradicting (**).
     
    This completes the proof of the lemma.
\end{proof}

\section{\pls-Completeness for \texorpdfstring{$k \geq 15$}{k >= 15}}
\label{sec:PLS-complete-proof}
We are now ready to prove a tight \pls-reduction from Max-Cut/Flip to TSP/\kopt.

\begin{theorem}
\label{thm:tight_reduction}
    There exists a tight \pls-reduction from Max-Cut/Flip with maximum degree five to TSP/\kopt for any $k \geq 15$.
\end{theorem}

\begin{proof}
	In order to show a tight \pls-reduction, we need to stipulate a \pls-reduction $(h,g)$ as defined in \cref{def:pls_complete} and a set $R$ as defined in \cref{def:tight_pls_complete}.

    For the function $h$, we use the reduction from an instance $(H, w)$ of Max-Cut to an instance $(G_{\infty}, c_{\infty})$ of TSP as described in \cref{sec:reduction} with the weight assignment of the non-edges as described in \cref{sec:no_infty_opt}.
    For the function $g$, we describe a mapping of a tour $\tau$ of $G_{\infty}$ to a cut $\gamma$ of $H$.
    If $\tau$ is a standard tour of $G$, then we map it to $\phi(\tau)$.
    Otherwise, we map it to an arbitrary cut of $H$ (e.g., the cut where all vertices are in the first set).
    
    We show that $(h,g)$ is indeed a \pls-reduction.
    It is easy to see that they can be computed in polynomial time.
    Further, by \cref{cor:no_nonstandard}, if a tour of $G_{\infty}$ is not a standard tour of $G$, it has to contain a non-edge, and by \cref{lem:no_infty_local_optima}, it cannot be a local optimum.
    Otherwise, \cref{lem:correspondence} implies that if $\tau$ is a locally optimal standard tour of $G$, then $\phi(\tau)$ is a locally optimal cut of $H$.
    All of the above show that $(h,g)$ is a \pls-reduction.
    
    For tightness, we define the set $R$ to be the set of all standard tours of $G$.
    We now show that $R$ satisfies the three conditions in \cref{def:tight_pls_complete}.
    As argued above, a locally optimal solution of $(G_{\infty}, c_{\infty})$ is a standard tour of $G$; hence, the first condition is satisfied.
    The second condition follows from \cref{lem:phi_bijection}.
    Lastly, note that a solution not in $R$ is a tour with a non-edge.
    Since the weight of the non-edge is larger than the total weight of any standard tour, an improving $k$-swap can only transform a standard tour to another standard tour.
    Combined with \cref{lem:correspondence}, the third condition then follows.
    This completes the proof of the tight \pls-reduction.
\end{proof}

The main result of our paper then immediately follows.

\plscomplete*
\begin{proof}
    The theorem follows from \cref{thm:tight_reduction} and the fact that Max-Cut/Flip is \pls-complete even when restricted to instances with maximum degree five~\cite{Elsaesser_Max_Cut_5}.
\end{proof}

\section{Concluding remark}

In this paper, we proved the \pls-completeness of TSP/\kopt for $k\ge 15$ via a tight \pls-reduction from Max-Cut/Flip with maximum degree five.
Two natural open problems remain.
\begin{itemize}
    \item Is TSP/\kopt \pls-complete for $k < 15$?
    It can be shown that there is no smaller ``strict gadget''; that is a gadget that has unique subtours (1)--(4), has no other subtours, and is either a $(3,2)$-, a $(2,3)$-, or a $(2,2)$-parity gadget in accordance to \cref{def:rx_ry_pg}.
    Therefore, a proof of \pls-completeness for lower values of $k$ might require a substantially different idea.
    
    \item Is TSP/\kopt tightly \pls-complete for some constant $k$? 
    Since our reduction is a tight \pls-reduction, if bounded degree Max-Cut/Flip is proven to be tightly \pls-complete in the future, a straightforward adaptation of our reduction (in particular, a slight change in the construction of the graph $\overrightarrow{H}$ in \cref{sec:equip}) will answer the second open question above positively.
\end{itemize}

\paragraph{Acknowledgement.} 
This work was initiated at the ``Discrete Optimization'' trimester program of the Hausdorff Institute of Mathematics, University of Bonn, Germany in 2021.
We would like to thank the organizers for providing excellent working conditions and inspiring atmosphere. 

\bibliographystyle{alpha}
\bibliography{refs}

\end{document}